\documentclass[a4paper,12pt]{article}

\usepackage{color}
\usepackage{amssymb, amsmath}
\usepackage{hyperref}
\usepackage{tikz}

\newcommand\EFFACE[1]{}

\newtheorem{theorem}{Theorem}
\newtheorem{proposition}[theorem]{Proposition}
\newtheorem{definition}[theorem]{Definition}
\newtheorem{lemma}[theorem]{Lemma}
\newtheorem{corollary}[theorem]{Corollary}
\newtheorem{conjecture}[theorem]{Conjecture}

\newtheorem{example}[theorem]{Example}

\newtheorem{observation}[theorem]{Observation}

\newenvironment{proof}{
\par
\noindent {\bf Proof.}\rm}{\mbox{}\hfill$\square$\par\vskip 3mm}

\newcommand\SOMMET[1]{\draw[fill=black] (#1) circle (2pt)}

\def\ZZZ{\mathbb{Z}}
\def\Aut{{\rm Aut}}
\def\Id{{\rm Id}}

\def\IS{\sigma} 

\def\p2{\frac{q}{2}}  

\makeatletter
\let\@fnsymbol\@arabic
\makeatother

\begin{document}


\title{On the Distinguishing Number of Cyclic Tournaments:
Towards the Albertson-Collins Conjecture}

\author{Kahina MESLEM~\thanks{LaROMaD, Faculty of mathematics, University of Sciences and Technology
Houari Boumediene (USTHB), B.P.~32 El-Alia, Bab-Ezzouar, 16111 Algiers, Algeria.}
\and \'Eric SOPENA~\thanks{Univ. Bordeaux, Bordeaux INP, CNRS, LaBRI, UMR5800, F-33400 Talence, France.}
}

\maketitle

\abstract{
A distinguishing $r$-labeling of a digraph $G$ is a mapping $\lambda$ from the set of vertices
of $G$ to the set of labels $\{1,\dots,r\}$ such that 
no nontrivial automorphism of $G$ preserves all the labels.
The distinguishing number $D(G)$ of $G$ is then the smallest $r$ for which $G$ admits a distinguishing $r$-labeling.
From a result of Gluck (David Gluck, Trivial set-stabilizers in finite permutation groups,
{\em Can. J. Math.} 35(1) (1983), 59--67),
it follows that $D(T)=2$ for every cyclic tournament~$T$ of (odd) order $2q+1\ge 3$.
Let $V(T)=\{0,\dots,2q\}$ for every such tournament.
Albertson and Collins conjectured in 1999
that the canonical 2-labeling $\lambda^*$, given by
$\lambda^*(i)=1$ if and only if $i\le q$, is distinguishing.

We prove that whenever one of the sub-tournaments of $T$ induced by the set of vertices $\{0,\dots,q\}$
or $\{q+1,\dots,2q\}$ is rigid, $T$ satisfies the Albertson-Collins Conjecture.
Using this property, we prove that several classes of cyclic tournaments satisfy the Albertson-Collins Conjecture.
Moreover, we also prove that every Paley tournament satisfies the Albertson-Collins Conjecture.
}

\medskip



\section{Introduction}

An \emph{$r$-labeling} of a graph or a digraph $G$ is a mapping $\lambda$ from the set of vertices $V(G)$
of $G$ to the set of labels $\{1,\dots,r\}$.
An automorphism $\phi$ of $G$ is \emph{$\lambda$-preserving} if $\lambda(\phi(u))=\lambda(u)$ for
every vertex $u$ of $G$.
An $r$-labeling  $\lambda$ of $G$ is \emph{distinguishing} if the only
$\lambda$-preserving automorphism of $G$ is the identity, that is, the labeling $\lambda$
breaks all the symmetries of $G$.
In~\cite{AC96}, Albertson and Collins introduce the \emph{distinguishing number} of $G$,
denoted $D(G)$, defined as the smallest $r$ for which $G$ admits a distinguishing $r$-labeling.

A digraph $G$ is \emph{rigid} (or \emph{asymmetric}) if the only automorphism of $G$ is the identity. Therefore,
$D(G)=1$ if and only if $G$ is rigid.

In the last decade, distinguishing numbers have been studied for several families of graphs, such
as trees~\cite{C06,WZ07},
interval graphs~\cite{C09},
planar graphs~\cite{ACD08},
hypercubes~\cite{BC04},
Cartesian products of graphs~\cite{A05,B13a,EIKPT16,FI08,IJK08,IK06,KZ07}, 
Kneser graphs~\cite{AB07,B13b},
infinite graphs~\cite{BD10,IKT07,SW14},
for instance.
However, distinguishing numbers of digraphs have been less frequently considered
since the paper of Albertson and Collins~\cite{AC99}
(see \cite{LNS10,LS12,L13}).

Symmetry breaking in tournaments is studied in~\cite{L13} by Lozano
who considers
other ways of distinguishing vertices, namely by means of determining sets~\cite{B06}
 (sometimes called fixing sets) or  resolving sets.

\medskip

We denote by $\ZZZ_n$ the group of residues modulo $n$.
Let $q$ be an integer, $q\ge 1$, and $S$ be a subset of $\ZZZ_{2q+1}\setminus\{0\}$ such that
for every $k\in\ZZZ_{2q+1}\setminus\{0\}$, $|S\cap\{k,2q+1-k\}|=1$.
The \emph{cyclic tournament} $T=T(2q+1;S)$ 
 is the tournament of order $2q+1$ defined by $V(T)=\ZZZ_{2q+1}$
and $ij$ is an arc in $T$ if and only if $j-i\in S$.
Cyclic tournaments are sometimes called \emph{circulant} or \emph{rotational} tournaments in the literature.

Let $\Gamma$ be a permutation group on a finite set $\Omega$.
A subset $R$ of $\Omega$ is said to be \emph{$\Gamma$-regular}
if its stabilizer $S(R)=\{g\in \Gamma\ |\ gR=R\}$ is trivial,
that is, $S(R)=\{\Id_{\Omega}\}$, where $\Id_{\Omega}\in\Gamma$ stands for the identity 
permutation acting on ${\Omega}$.
In~\cite{G83}, Gluck proves the following theorem (see also~\cite{M02}
for a simpler proof of this result).

\begin{theorem}[Gluck~\cite{G83}]
\label{th:Gluck}
Let $\Gamma$ be a permutation group of odd order on a finite set $\Omega$.
Then~$\Gamma$ has a regular subset in $\Omega$.
\end{theorem}

Theorem~\ref{th:Gluck} implies that $D(T)=2$ for every cyclic tournament $T=T(2q+1;S)$, $q\ge 1$.
Indeed, it follows from Gluck's Theorem that there exists a subset $R_T$ of $V(T)$ 
such that $\gamma(R_T)\neq R_T$
for every nontrivial automorphism $\gamma$ of $T$.
Therefore, the 2-labeling $\lambda_{R_T}$ defined by
$\lambda_{R_T}(i)=1$ if and only if $i\in R_T$ is clearly distinguishing.
On the other hand, $D(T)>1$ since $T$ is a cyclic tournament, and thus not rigid.

In~\cite{AC99}, Albertson and Collins study distinguishing 2-labelings of cyclic tournaments
and propose the following conjecture.

\begin{conjecture}[Albertson-Collins~\cite{AC99}]
For every cyclic tournament $T=T(2q+1;S)$, $q\ge 1$, the $2$-labeling $\lambda^*$, given
by $\lambda^*(i)=1$ if $0\le i\le q$ and $\lambda^*(i)=2$ otherwise, is distinguishing.
\label{conj:AC}
\end{conjecture}

In this paper, we prove that several classes of cyclic tournaments satisfy the Albertson-Collins Conjecture.
We first give some definitions, notation and basic results in Section~\ref{sec:basic}.
We then consider the so-called pseudo-cyclic tournaments in Section~\ref{sec:pseudo-cyclic} and 
prove our main results in Section~\ref{sec:main}.
We finally propose some directions for future work in Section~\ref{sec:discussion}.

\section{Preliminaries}
\label{sec:basic}

We denote by $V(G)$ and $E(G)$ the set of vertices and the set of arcs of a digraph $G$, respectively.
Let $G$ be a digraph and $u$ a vertex of $G$.
The \emph{outdegree} of $u$ in $G$, denoted $d^+_G(u)$, is the number of arcs in $E(G)$ of the form $uv$,
and the \emph{indegree} of $u$ in $G$, denoted $d^-_G(u)$, is the number of arcs in~$E(G)$ of the form $vu$.
The \emph{degree} of $u$, denoted $d_G(u)$, is then defined by $d_G(u)=d_G^+(u)+d_G^-(u)$.
If $uv$ is an arc in $G$, $u$ is an \emph{in-neighbour} of $v$ and $v$ is an \emph{out-neighbour} of $u$.
We denote by $N_G^+(u)$ and $N_G^-(u)$ the set of out-neighbours and the set of in-neighbours of $u$ in $G$,
respectively. Hence, $d^+_G(u)=|N_G^+(u)|$ and $d^-_G(u)=|N_G^-(u)|$.
Let $v$ and $w$ be two neighbours of $u$. We say that $v$ and $w$ \emph{agree on $u$}
if either both $v$ and $w$ are in-neighbours of $u$, or both $v$ and $w$ are out-neighbours of $u$,
and that $v$ and $w$ \emph{disagree on $u$} otherwise.

A digraph $G$ is \emph{regular} if all its vertices have the same in-degree and the same out-degree,
that is, $d_G^+(u)=d_G^+(v)$ and $d_G^-(u)=d_G^-(v)$ for every two vertices $u$ and $v$ in $G$.
For a subset $W$ of $V(G)$, we denote by $G[W]$ the sub-digraph of $G$
induced by $W$.

For any finite set $\Omega$, $\Id_\Omega$ denotes the identity permutation acting on $\Omega$.
Since the set $\Omega$ is always clear from the context, we simply write
$\Id$ instead of $\Id_\Omega$ in the following.

An \emph{automorphism} of a digraph $G$ is an arc-preserving permutation of its vertices, that is,
a one-to-one mapping $\phi:V(G)\rightarrow V(G)$ such that $\phi(u)\phi(v)$ is an arc in $G$ whenever $uv$
is an arc in~$G$.
The set of automorphisms of $G$ is denoted $\Aut(G)$.
The \emph{order} of an automorphism $\phi$ is the smallest integer $k>0$ for which $\phi^k=\Id$.
An automorphism $\phi$ of a digraph $G$ is \emph{nontrivial} if $\phi\neq \Id$. 
A vertex $u$ of $G$ is \emph{fixed} by $\phi$ if $\phi(u)=u$.
The \emph{orbit of $u$ with respect to $\phi$} is the set $\{u,\phi(u),\dots,\phi^{r-1}(u)\}$,
where $r$ is the \emph{order of $u$ with respect to $\phi$}, that is, the smallest positive 
integer for which $\phi^r(u)=u$ (note that the order of $u$ necessarily divides the order of~$\phi$).
An \emph{orbit of~$\phi$} is then an orbit of some vertex with respect to~$\phi$.
A tournament cannot admit an automorphism of order 2 (such an automorphism would interchange the ends
of some arc) and thus the automorphism group of a tournament has odd order, and
every orbit of an automorphism of a tournament contains either one or at least three elements.

The following observations will be useful later.

\begin{observation}
\label{obs:agree-fixpoint}
Let $T$ be a tournament and $u$, $v$ two vertices of $T$.
If $v$ is fixed by an automorphism $\phi\in\Aut(T)$ and the orbit of $u$ with respect to $\phi$
is of size at least~3, then all vertices in the orbit of $u$ agree on $v$.
\end{observation}

Indeed, if $u'$ belongs to the orbit of $u$ with respect to $\phi$, say $u'=\phi^s(u)$ for some integer $s>0$,
the result follows from the fact that $\phi^s(v)=v$ and $\phi^s$ is an automorphism of $T$.

\begin{observation}
\label{obs:agree-odd-orbits}
Let $T$ be a tournament and $\phi\in\Aut(T)$.
If $O_1$ and $O_2$ are two (not necessarily distinct) orbits of $\phi$ of size at least~3,
then, for every two vertices $u$ and $u'$ in $O_1$, 
$$|N_T^+(u)\cap O_2|=|N_T^+(u')\cap O_2| \ \ \ \mbox{and}\ \ \ |N_T^-(u)\cap O_2|=|N_T^-(u')\cap O_2|.$$
\end{observation}

To see that, suppose $\phi^s(u)=u'$ for some integer $s>0$, and let 
$N_T^+(u)\cap O_2=\{v_1,\dots,v_k\}$, $k>0$ (if $N_T^+(u)\cap O_2=\emptyset$, the result is obvious).
We then have $N_T^+(u')\cap O_2=\{\phi^s(v_1),\dots,\phi^s(v_k)\}$, which implies
$|N_T^+(u)\cap O_2|=|N_T^+(u')\cap O_2|$ since $\phi^s$ is an automorphism of $T$,
and thus $|N_T^-(u)\cap O_2|=|N_T^-(u')\cap O_2|$.

\begin{observation}
\label{obs:orbit-regular}
Let $T$ be a tournament and $\phi\in\Aut(T)$.
Every orbit of $\phi$ induces a regular sub-tournament of $T$.
\end{observation}

This directly follows from Observation~\ref{obs:agree-odd-orbits}, when considering the case $O_1=O_2$.

\medskip

The \emph{transitive tournament} of order $n$, denoted $TT_n$, is defined by 
$V(TT_n)=\{0,\dots,n-1\}$ and, for every $i,j\in V(TT_n)$, $ij$ is an arc whenever $i<j$.
Clearly, every transitive tournament is rigid (all its vertices have distinct indegrees)
and thus $D(TT_n)=1$ for every $n$.
The \emph{almost transitive tournament} of order $n$, denoted $TT^*_n$, is obtained
from $TT_n$ by reversing the arc from 0 to $n-1$. The tournament $TT^*_n$ is thus Hamiltonian.
It is also known that every almost transitive tournament of order at least~4 is rigid
(note that $TT^*_3$ is the directed 3-cycle, and this is not rigid).

\begin{observation}
\label{obs:almost-rigid}
For every $n>3$, the tournament $TT^*_n$ is rigid, and thus $D(TT^*_n)=1$.
\end{observation}

Indeed, every two vertices in $TT^*_n$ have distinct indegrees, and thus distinct outdegrees,
except the pairs $\{0,1\}$ and $\{n-2,n-1\}$,
as $d_{TT^*_n}^-(0)=d_{TT^*_n}^-(1)=1$ and $d_{TT^*_n}^+(n-2)=d_{TT^*_n}^+(n-1)=1$.
Since every automorphism of $TT^*_n$ has odd order, $TT^*_n$ is a rigid tournament.


Let $q$ be an integer, $q\ge 1$, and $S$ be a subset of $\ZZZ_{2q+1}\setminus\{0\}$ such that
for every $k\in\ZZZ_{2q+1}\setminus\{0\}$, $|S\cap\{k,2q+1-k\}|=1$. Let
$$S^+=S\cap\{1,\dots,q\},\ \ \ \mbox{and}\ \ \ S^-=\{1,\dots,q\}\setminus S^+.$$
We call the elements of $S^+$ the \emph{positive connectors} of the cyclic tournament $T(2q+1,S)$
and the elements of $S^-$ the \emph{negative connectors} of $T(2q+1,S)$.
Note that knowing either $S^+$ or $S^-$ is enough to determine $S$ since
$$S=S^+\cup\left\{-s\ :\ s\in\{1,\dots,q\}\setminus S^+\right\}
=\{1,\dots,q\}\setminus S^-\cup\{-s\ :\ s\in S^-\}.$$ 
Therefore, we preferably denote the cyclic tournament $T=T(2q+1;S)$
by $T=T(2q+1;S^-)$ whenever we deal with an explicit set $S$.
In that case, for every $i,j\in\ZZZ_{2q+1}$, $i<j$, $ij$ is an arc in $T=T(2q+1;S^-)$
if and only if 
either $j-i\le q$ and $j-i\notin S^-$,
or $j-i>q$ and $2q+1-j+i\in S^-$.

Note that any cyclic tournament $T=T(2q+1;S)$ is regular. More precisely, 
$d_T^-(u)=d_T^+(u)=q=|S|$ for every vertex $u\in T$.


The {\em converse} $T^c$ of a tournament $T$ is obtained from $T$ be reversing all the arcs.
Clearly, the converse of any cyclic tournament $T=T(2q+1;S)$ is the cyclic tournament
$T^c=T(2q+1;S^c)$ with $S^c=\{1,\dots,2q\}\setminus S$. Moreover, $T$ and $T^c$ are isomorphic,
via the mapping $\gamma$ defined by $\gamma(0)=0$ and $\gamma(i)=2q+1-i$ for every $i\in V(T)\setminus\{0\}$.

\section{Pseudo-cyclic tournaments}
\label{sec:pseudo-cyclic}

Let $T=T(2q+1;S)$ be a cyclic tournament and $i,j$ two vertices of $T$ with $i<j$.
We denote by $T_{i,j}$ the sub-tournament of $T$ induced
by the set of vertices $\{i,i+1,\dots,j\}$.
Note that the sub-tournament $T_{i,j}$ is not necessarily a cyclic tournament.

In the rest of this paper, we call the 2-labeling $\lambda^*$
defined in Conjecture~\ref{conj:AC} the \emph{canonical 2-labeling} of $T$.
The canonical 2-labeling $\lambda^*$ thus
assigns label~1 to vertices of $T_{0,q}$ and label~2 to vertices of $T_{q+1,2q}$.

Since the tournament $T=T(2q+1;S)$ is cyclic, the sub-tournaments $T_{q+1,2q}$ and $T_{0,q-1}$ are isomorphic.
Recall that the set $S$ is characterized by the set $S^-$ of its negative connectors.
In the following, we thus study (not necessarily cyclic) tournaments of the following form.

\begin{definition}
{\rm
Let $N$ be a subset of $\ZZZ_{q+1}\setminus\{0\}$, $q\ge 2$, whose elements are called \emph{negative connectors}.
The \emph{pseudo-cyclic} tournament $P=P(q;N)$ is the tournament of order $q+1$ defined by 
$V(P)=\ZZZ_{q+1}$ and $ij$, $i>j$, is an arc in $P$ if and only if $i-j\in N$.
}\end{definition}

Note that if $T=T(2q+1;S)$  is a cyclic tournament,
then $T_{0,q}\cong P(q;S^-)$ and $T_{q+1,2q}\cong T_{0,q-1}\cong P(q-1;S^-)$ (if $q\notin S^-$),
or $T_{q+1,2q}\cong T_{0,q-1}\cong P(q-1;S^-\setminus\{q\})$ (if $q\in S^-$).

We first prove that $T=T(2q+1;S)$ satisfies the Albertson-Collins Conjecture
 whenever $T_{0,q}$ or $T_{q+1,2q}$ is rigid.

\begin{proposition}
\label{prop:T1orT2isrigid}
Let $T=T(2q+1;S)$ be a cyclic tournament.
If $T_{0,q}$ is rigid or $T_{q+1,2q}$ is rigid, then the canonical $2$-labeling
$\lambda^*$ of $T$ is distinguishing.
\end{proposition}

\begin{proof}
Let $\phi$ be a $\lambda^*$-preserving automorphism of $T$, that is, $\lambda^*\circ\phi=\lambda^*$,
and let $\phi_1$ and $\phi_2$ denote the restriction of $\phi$ to $T_{0,q}$ and $T_{q+1,2q}$, respectively.
Since $\lambda^*\circ\phi=\lambda^*$, both $\phi_1$ and $\phi_2$ are automorphisms.
Moreover, since $T_{0,q}$ is rigid or $T_{q+1,2q}$ is rigid, we get $\phi_1=\Id$ or $\phi_2=\Id$.
We will prove that we necessarily have $\phi_1=\Id$ and $\phi_2=\Id$, which gives $\phi=\Id$
so that $\lambda^*$ is a distinguishing labeling of $T$.

Suppose first that $\phi_1=\Id$ and assume to the contrary that $\phi_2\neq \Id$.
Let $a_1\in\{q+1,\dots,2q\}$ be the ``smallest'' non-fixed vertex of $T_{q+1,2q}$.
Since $T$ is cyclic, we can assume without loss of generality $a_1=q+1$.
(If $a_1\neq q+1$, by using the ``shift'' automorphism $\alpha:i\mapsto i+a_1-q-1$,
the sub-tournaments $T_{0,q}$ and $T_{q+1,2q}$ are shifted to $T_{a_1-q-1,a_1-1}$
and $T_{a_1,a_1-q-2}$, respectively. The two restricted automorphisms $\phi_1$
and $\phi_2$ become $\phi'_1=\alpha\phi_1\alpha^{-1}$ and 
$\phi'_2=\alpha\phi_2\alpha^{-1}$, respectively, and we still have $\phi'_1=\Id$.)
Let $O_1$ denote the orbit of $a_1$ with respect to $\phi_2$.
Since the size of $O_1$ is at least~3 and the
sub-tournament $T[O_1]$ induced by $O_1$ is regular (Observation~\ref{obs:orbit-regular}),
$a_1$ necessarily belongs to a
cycle in $T[O_1]$ and thus to a 3-cycle in $T[O_1]$, say $a_1a_2a_3$.

Since $a_1a_2a_3$ is a 3-cycle, we get $\{a_2-a_1,a_3-a_2,a_1-a_3\}\subseteq S$.
If $a_2<a_3$, let $w=a_1-a_3+a_2$, so that $wa_1$ is an arc of $T$. We then have $w\in V(T_{0,q})$ and thus, since $w$ is fixed
by $\phi$, $wa_1$ and $wa_2$ are both arcs of $T$, a contradiction since $a_2-w=a_3-a_1$
and $a_3-a_1\notin S$.
If $a_2>a_3$, we get a similar contradiction by considering the vertex $w=a_1+a_3-a_2$:
$a_1w$ is an arc of $T$, again $w\in V(T_{0,q})$, which implies that $a_1w$ and $a_3w$ are both arcs of $T$,
a contradiction since $w-a_3=a_1-a_2\notin S$.

The case $\phi_2=\Id$ is similar.
\end{proof}

Note that Proposition~\ref{prop:T1orT2isrigid} implies that whenever
$T_{0,q}$ or $T_{q+1,2q}$ is rigid, both sets $V(T_{0,q})$ and $V(T_{q+1,2q})$ are $\Aut(T)$-regular.
It should also be noticed that the condition in Proposition~\ref{prop:T1orT2isrigid}
is sufficient for a cyclic tournament to satisfy the Albertson-Collins Conjecture,
but not necessary, as shown by the following example.

\begin{example}
{\rm
Consider the cyclic tournament $T=T(13;\{2,5,6\})$.
The automorphism group of $T$ only contains rotations (that is, mappings $\phi:i\mapsto i+b$, $0\le b\le 12$),
so that the canonical 2-labeling $\lambda^*$ is clearly distinguishing, and thus $T$
satisfies the Albertson-Collins Conjecture.
However, none of the sub-tournaments $T_{0,6}$ and $T_{7,12}$ is rigid.
The sub-tournament $T_{0,6}\cong P(6;\{2,5,6\})$ admits
an automorphism of order 3, namely $\phi'=(0,3,6)$,
and the
sub-tournament $T_{7,12}\cong P(5;\{2,5\})$ also admits an automorphism of order 3,
namely $\phi''=(7,8,9)(10,11,12)$. 
}
\label{ex:T13}
\end{example}

Moreover, we have the following.

\begin{proposition}
Let $T=T(2q+1;S)$ be a cyclic tournament and $\lambda^*$ be the canonical $2$-labeling of $T$.
If $\phi$ is a nontrivial $\lambda^*$-preserving automorphism of $T$, 
then $\phi$ has at least two orbits, both in $T_{0,q}$ and in $T_{q+1,2q}$.
\label{prop:no-single-orbit}
\end{proposition}

\begin{proof}
Suppose first that $\phi$ has only a single orbit in $T_{0,q}$,
 which implies that $q$ is even.
No vertex $a$ of $T_{q+1,2q}$ can be fixed by $\phi$, since otherwise,
by Observation~\ref{obs:agree-fixpoint}, all vertices in $T_{0,q}$
would agree on $a$, which would imply
either $d^-_T(a)\ge q+1$ or $d^+_T(a)\ge q+1$, in contradiction with the definition of $T$.
Therefore, $\phi$ has an even number of orbits 
in $T_{q+1,2q}$, of respective odd sizes $q_1,\dots,q_{2\ell}$,
$\ell\ge 1$,
with $q_i\ge 3$ for every $i$, $1\le i\le 2\ell$, and $\sum_{i=1}^{2\ell} q_i=q$.
Since $q_1<q$, the automorphism $\phi_1=\phi^{q_1}$ is $\lambda^*$-preserving and has a single
orbit in $T_{0,q}$. But $\phi_1$ fixes $q_1$ vertices of $T_{q+1,2q}$, contradicting
the fact that any $\lambda^*$-preserving automorphism for which $T_{0,q}$
has a single orbit fixes no vertices of $T_{q+1,2q}$.

Suppose now that $\phi$ has only a single orbit in $T_{q+1,2q}$, which implies that $q$ is odd.
As before, we claim that no vertex of $T_{0,q}$ is fixed by $\phi$. 
Assume to the contrary that $a$ is such a vertex. 
As above, by Observation~\ref{obs:agree-fixpoint}, we get
$d^-_{T_{0,q}}(a)\ge q$ or $d^+_{T_{0,q}}(a)\ge q$, 
and thus $d^-_{T_{0,q}}(a)=d^+_{T_{0,q}}(a)=q$ according to the definition of $T$.
Since $V(T_{q+1,2q})$ is either the set of in-neighbours or the set of out-neighbours
of $a$, we get that the elements of $S^-$ or of $S^+$, respectively, are consecutive integers, and thus
 every two vertices in $T_{0,q}$
have distinct indegrees in $T_{0,q}$.
Hence, the sub-tournament $T_{0,q}$ is rigid,
which implies $\phi=\Id$ by Proposition~\ref{prop:T1orT2isrigid}, a contradiction.
Therefore, $\phi$ has an even number of orbits in $T_{0,q}$, each of size at least~3, 
and a contradiction arises as in the previous case.
\end{proof}

Let $P=P(q;N)$ be a pseudo-cyclic tournament.
The \emph{indegree sequence} of~$P$ is the sequence defined by
$$\IS(P)=(d_P^-(0),\dots,d_P^-(q)).$$
The value of $d_P^-(i)$ for any vertex $i$ of~$P$ is given by the following.

\begin{proposition}
The indegree $d_P^-(i)$ of any vertex $i$ of the pseudo-cyclic tournament $P=P(q;N)$, $0\le i\le q$,
is given by 
$$d_P^-(i)=\left\{\begin{array}{ll}
              i+|N\cap\{i+1,\dots,q-i\}| & \hspace{3mm} \mbox{if}\  0\leq i\leq\left\lfloor{\p2}\right\rfloor, \\
              q-d_P^-(q-i) & \hspace{3mm}  \mbox{otherwise}.\\
            \end{array}
\right.$$
\label{prop:indegree}
\end{proposition}

\begin{proof}
By definition of~$P$, there is an arc from $j$ to $i$, $j>i$, if and only if $j-i\in N$.
Hence,
$$d_P^-(i)=(i-|N\cap\{1,\dots,i\}|)+|N\cap\{1,\dots,q-i\}|.$$
Therefore,
we get
$$d_P^-(i)=i+|N\cap\{i+1,\dots,q-i\}|$$
if $0\leq i\leq\left\lfloor{\p2}\right\rfloor$ (with, in particular,  $d_P^-({\p2})={\p2}$ if $q$ is even), and
$$d_P^-(i)=i-|N\cap\{q-i+1,\dots,i\}|=q-d_P^-(q-i)$$
if $\left\lceil{\p2}\right\rceil\le i\le q$.
\end{proof}

From Proposition~\ref{prop:indegree}, we get that for every 
pseudo-cyclic tournament $P=P(q;N)$, $d_P^-(0)=|N|$ and $d_P^-(q)=q-|N|$,
and that the indegree sequence $\IS(P)$ admits a central symmetry.
Moreover, the difference of the indegrees of any two consecutive vertices
is either $-1$, $0$ or $1$.

\begin{proposition}
For any two consecutive vertices $i$ and $i+1$ of the pseudo-cyclic tournament $P=P(q;N)$, 
$0\le i\le q-1$, $d_P^-(i+1)-d_P^-(i)\in\{-1,0,1\}$. More precisely, we have
\begin{enumerate}
\item $d_P^-(i+1)-d_P^-(i)=1 - |N\cap\{i+1,q-i\}|$,
\item if $q$ is odd, 
$d_P^-\left(\left\lceil{\p2}\right\rceil\right)-d_P^-\left(\left\lfloor{\p2}\right\rfloor\right)=
1-2\left|N\cap\left\{\left\lfloor{\p2}\right\rfloor+1\right\}\right|$.
\end{enumerate}
\label{prop:indegree-difference}
\end{proposition}

\begin{proof}
By Proposition~\ref{prop:indegree},
if $0\le i<\left\lfloor{\p2}\right\rfloor$ we have  
$$
\begin{array}{lll}
d_P^-(i+1)-d_P^-(i) & = & i+1+ |N\cap\{i+2,\dots,q-i-1\}|\\
                    &   & -\ i - |N\cap\{i+1,\dots,q-i\}|\\
                    & = & 1 - |N\cap\{i+1,q-i\}|.
\end{array}$$
Since $0\le |N\cap\{i+1,q-i\}|\le 2$, we get $d_P^-(i+1)-d_P^-(i)\in\{-1,0,1\}$.

If $\left\lceil{\p2}\right\rceil\le i< p$, then, by Proposition~\ref{prop:indegree}, we get
$$
\begin{array}{lll}
d_P^-(i+1)-d_P^-(i)
  & = & q-d^-(q-i-1)-q+d^-(q-i)\\
  & = & d^-(q-i)-d^-(q-i-1)\\
  & = & 1 - |N\cap\{q-i-1+1,q-q+i+1\}|\\
  & = & 1 - |N\cap\{i+1,q-i\}|.
\end{array}$$

Finally, if $q$ is odd, we have 
$$
\begin{array}{lll}
d_P^-\left(\left\lceil{\p2}\right\rceil\right)-d_P^-\left(\left\lfloor{\p2}\right\rfloor\right)
  & = & q-(\left\lfloor{\p2}\right\rfloor + \left|N\cap\left\{\left\lfloor{\p2}\right\rfloor+1\right\}\right|)
  -(\left\lfloor{\p2}\right\rfloor + \left|N\cap\left\{\left\lfloor{\p2}\right\rfloor+1\right\}\right|)\\
                    & = & 1-2\left|N\cap\left\{\left\lfloor{\p2}\right\rfloor+1\right\}\right|.
\end{array}$$
This completes the proof.
\end{proof}

Note that if $q$ is odd, the value of 
$d_P^-\left(\left\lceil{\p2}\right\rceil\right)-d_P^-\left(\left\lfloor{\p2}\right\rfloor\right)$
is either $1$ or $-1$.

We can then define three types of vertices as follows.

\begin{definition}
{\rm
Let $P=P(q;N)$ be a pseudo-cyclic tournament. A vertex $i\in V(P)$, $0\le i\le q-1$,
is said to be
\begin{enumerate}
\item an \emph{ascent-vertex} if $d^-_P(i+1)=d^-_P(i)+1$, that is, $N\cap\{i+1,q-i\}=\emptyset$,
\item a \emph{descent-vertex} if $d^-_P(i+1)=d^-_P(i)-1$, that is, $\{i+1,q-i\}\subseteq N$, or
\item a \emph{plateau-vertex} if $d^-_P(i+1)=d^-_P(i)$, that is, $|N\cap\{i+1,q-i\}|=1$.
\end{enumerate}
}\end{definition}

A pseudo-cyclic tournament $P=P(q;N)$ is cyclic if and only if 
all its vertices have the same indegree, which implies $q$ is even and $|N|={\p2}$,
and the set $N$ is such that $i\in N$ if and only $q+1-i\notin N$.
In other words, we have the following.

\begin{observation}
A pseudo-cyclic tournament $P=P(q;N)$ is cyclic if and only if every vertex $i\in V(P)$,
$0\le i\le q-1$, is a plateau-vertex.
\label{obs:pseudo-is-cyclic}
\end{observation}

Moreover, we also have the following.

\begin{observation}
Let $P^c=P(q;N^c)$ be the converse pseudo-cyclic tournament of $P=P(q;N)$,
that is, $N^c=\{1,\dots,q\}\setminus N$.
If $i$ is an ascent-vertex (resp. a descent-vertex, a plateau-vertex) in $P$, $0\le i\le q-1$,
then $i$ is a descent-vertex (resp. an ascent-vertex, a plateau-vertex) in $P^c$.
\label{obs:converse-pseudo}
\end{observation}

This observation directly follows from the fact that $d^-_{P^c}(i)=q-d^-_P(i)$ for every vertex $i$.

We denote by $\alpha(P)$, $\delta(P)$ and $\pi(P)$ the number of 
ascent-vertices, descent-vertices and plateau-vertices
in $\IS(P)$, respectively, so that $\alpha(P)+\delta(P)+\pi(P)=q$.
Moreover, according to the three above types of vertices, we define three types of subsequences of
the indegree sequence $\IS(P)$.

\begin{definition}
{\rm
Let $P=P(q;N)$ be a pseudo-cyclic tournament.
A sequence of $k\ge 2$ consecutive vertices $(i,\dots,i+k-1)$, $i\le n-k+1$, of~$P$ is called
\begin{enumerate}
\item an \emph{ascent of size $k$} if $d_P^-(i+j+1)=d_P^-(i+j)+1$ for every $j$, $0\le j\le k-2$,
\item a \emph{descent of size $k$} if $d_P^-(i+j+1)=d_P^-(i+j)-1$ for every $j$, $0\le j\le k-2$,
\item a \emph{plateau of size $k$} if $d_P^-(i+j+1)=d_P^-(i+j)$ for every $j$, $0\le j\le k-2$.
\end{enumerate}
}\end{definition}

Note here that an ascent, a descent or a plateau of size $k$ contains $k-1$ ascent-vertices, descent-vertices
or plateau-vertices, respectively.

\begin{figure}
\begin{center}
\begin{tikzpicture}[domain=0:11,x=1.3cm,y=1.3cm]
\node[draw,circle] (v0)at(0,0){0};
\node[draw,circle] (v1)at(1,0){1};
\node[draw,circle] (v2)at(2,0){2};
\node[draw,circle] (v3)at(3,0){3};
\node[draw,circle] (v4)at(4,0){4};
\node[draw,circle] (v5)at(5,0){5};
\node[draw,circle] (v6)at(6,0){6};
\node[draw,circle] (v7)at(7,0){7};
\node[draw,circle] (v8)at(8,0){8};
\draw[->,thick] (v2) to[bend right] (v0);
\draw[->,thick] (v3) to[bend right] (v1);
\draw[->,thick] (v4) to[bend right] (v2);
\draw[->,thick] (v5) to[bend right] (v3);
\draw[->,thick] (v6) to[bend right] (v4);
\draw[->,thick] (v7) to[bend right] (v5);
\draw[->,thick] (v8) to[bend right] (v6);
\draw[->,thick] (v4) to[bend right] (0.15,0.25); 
\draw[->,thick] (v5) to[bend right] (1.15,0.25); 
\draw[->,thick] (v6) to[bend right] (2.15,0.25); 
\draw[->,thick] (v7) to[bend right] (3.15,0.25); 
\draw[->,thick] (v8) to[bend right] (4.15,0.25); 
\draw[->,thick] (v5) to[bend left] (v0);
\draw[->,thick] (v6) to[bend left] (v1);
\draw[->,thick] (v7) to[bend left] (v2);
\draw[->,thick] (v8) to[bend left] (v3);
\node[below]at(4,-1){(a) The pseudo-cyclic tournament $P(8;\{2,4,5\})$};
 
\end{tikzpicture}
\vskip 1cm
\begin{tikzpicture}[domain=0:11,x=1cm,y=1cm]
\draw[very thick,->] (-0.5,0) -- ++(5,0);
\draw[very thick,->] (0,-0.5) -- ++(0,3.5);
\node[right]at(4.5,0){$i$};
\node[above]at(0,3){$d^-_P(i)$};
\SOMMET{0,1.5}; \SOMMET{0.5,2}; \SOMMET{1,2}; \SOMMET{1.5,2.5}; \SOMMET{2,2}; \SOMMET{2.5,1.5}; \SOMMET{3,2}; \SOMMET{3.5,2}; \SOMMET{4,2.5}; 
\draw[very thick] (0,1.5)--(0.5,2);
\draw[very thick](0.5,2)--(1,2);
\draw[very thick](1,2)--(1.5,2.5);
\draw[very thick](1.5,2.5)--(2,2);
\draw[very thick](2,2)--(2.5,1.5);
\draw[very thick](2.5,1.5)--(3,2);
\draw[very thick](3,2)--(3.5,2);
\draw[very thick](3.5,2)--(4,2.5);
\draw[thick,dashed] (2,0)--(2,2);
\draw[thick,dashed] (1,0)--(1,2);
\draw[thick,dashed] (3,0)--(3,2);
\draw[thick,dashed] (4,0)--(4,2.5);
\draw[thick,dotted](0,1.5)--(4,1.5);
\draw[thick,dotted](0,2.5)--(4,2.5);
\node[below] at(1,0){2};
\node[below] at(2,0){4};
\node[below] at(3,0){6};
\node[below] at(4,0){8};
\node[left]at(0,1.5){3};
\node[left]at(0,2.5){5};
\node[below]at(2,-1){(b) The indegree path of $P(8;\{2,4,5\})$};
\end{tikzpicture}
\caption{The pseudo-cyclic tournament $P(8;\{2,4,5\})$ and its indegree path}
\label{fig:ex1}
\end{center}
\end{figure}

\begin{example}
{\rm The pseudo-cyclic tournament $P=P(8;\{2,4,5\})$ is depicted in Figure~\ref{fig:ex1}(a)
(only arcs corresponding to negative connectors are drawn, every missing arc is thus directed from
left to right).
The indegree sequence of~$P$ is given by
$$\IS(P)=(3,4,4,5,4,3,4,4,5)$$
and is represented by the (centrally symmetric) \emph{indegree path} of~$P$ depicted in Figure~\ref{fig:ex1}(b).
We then have $\pi(P)=2$, $\delta(P)=2$ and $\alpha(P)=4$.
}\end{example}

The number of descent-vertices and the number of plateau-vertices are related
to the cardinality of the set $N$ as follows.

\begin{proposition}
For every pseudo-cyclic tournament  $P=P(q;N)$,
$$\delta(P)+\frac{1}{2}\pi(P)=|N|.$$
\label{prop:cardinality-of-N}
\end{proposition}

\begin{proof}
If $N=\emptyset$, then $\IS(P)=(0,1,\dots,q)$ ($P$ is transitive), 
so that $\delta(P)=\pi(P)=0$ and we are done.

Otherwise, let $N=\{a_1,\dots,a_r\}$, so that $|N|=r\ge 1$.
We claim that each $a_i$ generates either a descent-vertex or two
plateau-vertices.
Indeed, for each $a_i\in N$, either
$q+1-a_i\not\in N$, which implies that both vertices $a_i-1$ and $q-a_i$ are plateau-vertices,
 or $q+1-a_i\in N$, which implies that
$a_i-1$ is a descent-vertex.
We thus get $|N|=\delta(P)+\frac{1}{2}\pi(P)$, as required.
\end{proof}

The following proposition shows that vertices with same indegree in
 a pseudo-cyclic tournament are necessarily ``not too far'' from each other.

\begin{proposition}
Let $P=P(q;N)$ be a pseudo-cyclic tournament.
If $i$ and $j$ are two vertices of~$P$ with $i<j$ and $d_P^-(i)=d_P^-(j)$, then
\begin{enumerate}
\item if $0\le i<j\le\left\lfloor{\p2}\right\rfloor$ or $\left\lceil{\p2}\right\rceil\le i<j\le q$, then
$j\le i+|N|$,
\item if $0\le i\le\left\lfloor{\p2}\right\rfloor<j\le q$, then
$j\le i+2|N|$.
\end{enumerate}
\label{prop:closeto}
\end{proposition}

\begin{proof}
Let $i$ and $j$ be such that $0\le i<j\le\left\lfloor{\p2}\right\rfloor$ and $j>i+|N|$.
By Proposition~\ref{prop:indegree}, we have
$$
\begin{array}{lll}
d_P^-(j)-d_P^-(i) & = & j+|N\cap\{j+1,\dots,q-j\}|) - i - |N\cap\{i+1,\dots,q-i\}|)\\
                    & = & j - i - |N\cap\left(\{i+1,j\}\cup\{q-j+1,q-i\}\right)|\\
                    & > & |N|-|N|=0.
\end{array}$$
Therefore, $d_P^-(i)\neq d_P^-(j)$.
Since the indegree path of~$P$ is centrally symmetric, the result also holds when
$\left\lceil{\p2}\right\rceil\le i<j\le q$.

Suppose now that $0\le i\le\left\lfloor{\p2}\right\rfloor<j\le q$ and $j>i+2|N|$.
By Proposition~\ref{prop:indegree}, we have
$$
\begin{array}{lll}
d_P^-(j)-d_P^-(i) & = & j-|N\cap\{q-j+1,\dots,j\}|) - i - |N\cap\{i+1,\dots,q-i\}|)\\
                    & = & j - i - (|N\cap\{q-j+1,\dots,j\}| + |N\cap\{i+1,\dots,q-i\}|)\\
                    & > & 2|N|-(|N|+|N|)=0.
\end{array}$$
Again, we get $d_P^-(i)\neq d_P^-(j)$.
\end{proof}

Finally, the following property will be useful in the sequel.

\begin{proposition}
Let $P=P(q;N)$ be a pseudo-cyclic tournament.
If a vertex $i$ is fixed by every automorphism of~$P$,
then the vertex $q-i$ is also fixed by every automorphism of~$P$.
\label{prop:fixed-by-symmetry}
\end{proposition}

\begin{proof}
Suppose to the contrary that $i$ is fixed by every automorphism of~$P$
and that there exists an automorphism $\phi$ of~$P$ such that $\phi(q-i)\neq q-i$.

Let $\phi^*:V(P)\longrightarrow V(P)$ be the mapping defined by $\phi^*(i)=q-\phi(q-i)$
for every $i\in V(P)$. 

We claim that $\phi^*\in \Aut(P)$.
Let $ij$ be an arc in $P$ with $i>j$, that is, $i-j\in N$.
We then have 
$$\phi^*(i)-\phi^*(j)=q-\phi(q-i)-q+\phi(q-j)=\phi(q-j)-\phi(q-i).$$
Since $q-j>q-i$ and $q-j-q+i=i-j\in N$, 
$(q-j)(q-i)$ is an arc in $P$.
Since $\phi\in \Aut(P)$, $\phi(q-j)\phi(q-i)$ is an arc in $P$.

Suppose first that $\phi(q-j)>\phi(q-i)$, which implies $\phi(q-j)-\phi(q-i)\in N$.
In that case $\phi^*(j)=q-\phi(q-j)<q-\phi(q-i)=\phi^*(i)$ and,
since $\phi^*(i)-\phi^*(j)=\phi(q-j)-\phi(q-i)\in N$,
$\phi^*(i)\phi^*(j)$ is an arc in $P$.

Suppose now that $\phi(q-j)<\phi(q-i)$, which implies $\phi(q-i)-\phi(q-j)\notin N$.
In that case $\phi^*(i)=q-\phi(q-i)<q-\phi(q-j)=\phi^*(j)$ and,
since $\phi^*(j)-\phi^*(i)=\phi(q-i)-\phi(q-j)\notin N$,
$\phi^*(i)\phi^*(j)$ is an arc in $P$.

The case of an arc $ij$ with $i<j$ is similar, and thus $\phi^*\in \Aut(P)$.

We then have $\phi^*(i)=q-\phi(q-i)\neq i$, contradicting the fact that $i$ is fixed by
every automorphism of~$P$.
\end{proof}

This immediately gives the following.

\begin{corollary}
\label{cor:half-is-fixed}
Let $P=P(q;N)$ be a pseudo-cyclic tournament.
If all vertices $i\le\left\lfloor{\p2}\right\rfloor$ of~$P$ are fixed by every automorphism of~$P$, then $P$ is rigid.
\end{corollary}

\section{Cyclic tournaments satisfying the Albertson-Collins\\ Conjecture}
\label{sec:main}

In this section, we prove that several classes of cyclic tournaments
$T(2q+1;S)$ satisfy the Albertson-Collins Conjecture.
We first propose a few sufficient conditions for a cyclic tournament
to satisfy the Albertson-Collins Conjecture and then consider
several specific classes of cyclic tournaments, depending on the structure
of the set $S$ of connectors.

\subsection{Simple sufficient conditions}
\label{subsec:sufficient-conditions}

By Proposition~\ref{prop:no-single-orbit}, for every nontrivial $\lambda^*$-preserving
automorphism $\phi$ of $T=T(2q+1;S)$, none of $T_{0,q}$ and $T_{q+1,2q}$ has a single
orbit with respect to $\phi$. Using this property, we can prove the following result.

\begin{lemma}
The cyclic tournament $T=T(2q+1;S)$ satisfies the Albertson-Collins Conjecture whenever
one of the following conditions is satisfied.
\begin{enumerate}
\item $\Aut(T)$ is of order $2q+1$,
\item $q$ is even, $|S^-|={\p2}$, $q+1-s\notin S^-$ for every $s\in S^-$, and $|\Aut(T_{0,q})|=q+1$,
\item $q$ is odd, $|S^-\setminus\{p\}|=\frac{q-1}{2}$, $q-s\notin S^-$ for every $s\in S^-$, and $|\Aut(T_{q+1,2q})|=q$.
\end{enumerate}
\label{lem:sufficient}
\end{lemma}

\begin{proof}
Recall that for every cyclic tournament $T_n$ of order $n$, the $n$ rotations of $T_n$
form a subgroup of $\Aut(T_n)$.

In Case~1, $\Aut(T)$ contains only rotations, so that $\lambda^*$ is clearly distinguishing.
In Case~2, $T_{0,q}$ is a cyclic tournament and every automorphism of $T_{0,q}$ is
a rotation. Therefore, the restriction $\phi_1$ to $T_{0,q}$
of any $\lambda^*$-preserving automorphism $\phi$ of $T$ is a rotation and thus $\phi_1=\Id$.
The result then follows by Proposition~\ref{prop:no-single-orbit}.
Case~3 is similar, using $T_{q+1,2q}$ instead of $T_{0,q}$.
\end{proof}

Let $T$ be a tournament. We denote by $V_d(T)$, $d\ge 0$,
the set of vertices of $T$ with indegree $d$. 
Since any automorphism maps every vertex to a vertex with same degree, we
have the following result.

\begin{lemma}
Let $P=P(n;N)$ be a pseudo-cyclic tournament.
If the sub-tournament $P[V_d(P)]$ of~$P$, induced by $V_d(T)$, is rigid for every $d$, $d\ge 0$, then $P$ is rigid.
\label{lem:Vd-rigid}
\end{lemma}

This lemma directly gives the following results.

\begin{theorem}
Let $T=T(2q+1;S)$ be a cyclic tournament, $P_1=T_{0,q}$ and $P_2=T_{q+1,2q}$. 
If $P_1[V_d(P_1)]$ is rigid for every $d$, 
or $P_2[V_d(P_2)]$ is rigid for every $d$, $d\ge 0$,
then $T$ satisfies the Albertson-Collins Conjecture.
\label{th:Vd-rigid}
\end{theorem}

\begin{proof}
The result directly follows from Lemma~\ref{lem:Vd-rigid} and Proposition~\ref{prop:T1orT2isrigid}.
\end{proof}

\begin{theorem}
Let $T=T(2q+1;S)$ be a cyclic tournament with $S^-\neq\emptyset$ and let $\min(S^-)$
denote the smallest negative connector of $S$.
If $\min(S^-)> 2|S^-|$, then $T$ satisfies the Albertson-Collins Conjecture.
\label{th:minN}
\end{theorem}

\begin{proof}
Consider the pseudo-cyclic sub-tournament $P=T_{0,q}$.
By Proposition~\ref{prop:closeto}, we know that if two vertices $i$ and $j$, $i<j$,
have the same indegree in $P$, then $j-i\le 2|S^-|$.
Since $\min(S^-)> 2|S^-|$, we get that
the sub-tournament $P[V_d(P)]$ is transitive,
and thus rigid, for every $d$, $d\ge 0$.
The result then follows by Theorem~\ref{th:Vd-rigid}.
\end{proof}

\subsection{Cyclic tournaments with few negative connectors}
\label{subsec:few}

We consider in this subsection the case of cyclic tournaments with at most two negative connectors.
We prove that every such tournament satisfies the Albertson-Collins Conjecture.

The proof of our result is by case analysis. Each of the forthcoming figures 
(including figures in the next subsection) is intended
to illustrate the indegree path associated with a specific case and, in particular, 
should allow the reader to easily identify
the sets of vertices having the same indegree (recall that every orbit of an automorphism
only contains vertices having the same indegree).

\begin{theorem}
If $T=T(2q+1;S)$ is a cyclic tournament with $q\ge 1$ and $|S^-|\le 2$,
then $T$ satisfies the Albertson-Collins Conjecture.
\label{th:S-atmost2}
\end{theorem}

\begin{proof}
The result is obvious if $q=1$ since, in that case, $T$ is the directed 3-cycle.
We thus assume $q\ge 2$.
We will prove that one of the two sub-tournaments $P_1=T_{0,q}$ or $P_2=T_{q+1,2q}$
is rigid, so that the result follows by Proposition~\ref{prop:T1orT2isrigid}.

If $S^-=\emptyset$, then $P_1$ is transitive and thus rigid.

If $|S^-|=1$, say $S^-=\{a\}$, then, by Proposition~\ref{prop:indegree}, the indegree
of any vertex $i$ of $P_1$ is as follows.
$$d_{P_1}^-(i)=\left\{\begin{array}{ll}
              i & \hspace{3mm} \mbox{if}\  a\notin\{i+1,\dots,q-i\}\ \mbox{and } 0\leq i\leq\left\lfloor{\p2}\right\rfloor, \\
              i+1 & \hspace{3mm} \mbox{if}\  a\in\{i+1,\dots,q-i\}\ \mbox{and } 0\leq i\leq\left\lfloor{\p2}\right\rfloor, \\
              q-d_{P_1}^-(q-i) & \hspace{3mm}  \mbox{otherwise}.\\
            \end{array}
\right.$$
Therefore, $d_{P_1}^-(i)=d_{P_1}^-(j)$, $i<j\le\left\lfloor{\p2}\right\rfloor$, if and only if $j=i+1$ and $j\in\{a,q+1-a\}$.
If $a\neq\left\lfloor{\p2}\right\rfloor$, then the number of vertices in $P_1$
with degree $d$ is at most 2 for every $d$. In that case, $P_1$ satisfies the condition
of Lemma~\ref{lem:Vd-rigid} and is thus rigid.
If $a=\left\lfloor{\p2}\right\rfloor$ the same property holds except if $q$ is
even (in that case, we have $|V_{{\p2}}(P_1)|=3$ and $|V_d(P_1)|=1$ for every
$d\neq{\p2}$). But then, $P_2$ satisfies the condition
of Lemma~\ref{lem:Vd-rigid} (since $q-1$ is odd) and we are done too.

Assume now that $|S^-|=2$ and let $S^-=\{a_1,a_2\}$.
If $q=2$, then $P_2$ has order 2 and is thus rigid.
If $q=3$, then either $S^-=\{1,2\}$, $S^-=\{1,3\}$ or $S^-=\{2,3\}$.
It is then easily checked that $|V_1(P_1)|=|V_2(P_1)|=2$ in all cases, so that $P_1$
satisfies the condition of Lemma~\ref{lem:Vd-rigid} and we are done.

Suppose thus that $q>3$ and, moreover, that $q$ is even (if $q$ is odd we simply
consider $P_2$ instead of $P_1$).
By Proposition~\ref{prop:cardinality-of-N}, we know that
$\delta(P_1)+\frac{1}{2}\pi(P_1)=2$. We thus have three  cases to consider.

\begin{figure}
\begin{center}
\begin{tikzpicture}[domain=0:11,x=1cm,y=1cm]
\draw[very thick,->] (-0.5,0) -- ++(5,0);
\draw[very thick,->] (0,-0.5) -- ++(0,3.5);
\node[right]at(4.5,0){$i$};
\node[above]at(0,3){$d^-_{P_1}(i)$};
\SOMMET{0,0.5}; \SOMMET{1,1.5}; \SOMMET{1.5,2}; \SOMMET{2,1.5}; \SOMMET{2.5,1}; \SOMMET{3,1.5}; \SOMMET{4,2.5};
\draw[very thick,dotted] (0,0.5)--(1,1.5);
\draw[very thick](1,1.5)--(1.5,2);
\draw[very thick](1.5,2)--(2,1.5);
\draw[very thick](2,1.5)--(2.5,1);
\draw[very thick](2.5,1)--(3,1.5);
\draw[very thick,dotted](3,1.5)--(4,2.5);
\draw[thick,dashed] (2,0)--(2,1.5);
\draw[thick,dashed] (1.5,0)--(1.5,2);
\draw[thick,dashed] (0,2.5)--(4,2.5);
\node[below] at(2,0){$d_2$};
\node[below] at(1.5,0){$d_1$};
\node[left]at(0,0.5){2};
\node[left]at(0,2.5){$q-2$};
\node[below]at(1.75,-1){(a) $\delta(P_1)=2$, $\pi(P_1)=0$, $d_2-d_1=1$};
\end{tikzpicture}
\hskip 1cm
\begin{tikzpicture}[domain=0:11,x=1cm,y=1cm]
\draw[very thick,->] (-0.5,0) -- ++(5.5,0);
\draw[very thick,->] (0,-0.5) -- ++(0,4);
\node[right]at(5,0){$i$};
\node[above]at(0,3.5){$d^-_{P_1}(i)$};
\SOMMET{0,0.5}; \SOMMET{1,1.5}; \SOMMET{1.5,1}; \SOMMET{2,1.5}; \SOMMET{2.5,2}; \SOMMET{3,2.5}; \SOMMET{3.5,2}; \SOMMET{4.5,3};
\draw[very thick,dotted] (0,0.5)--(1,1.5);
\draw[very thick](1,1.5)--(1.5,1);
\draw[very thick](1.5,1)--(2,1.5);
\draw[very thick,dotted](2,1.5)--(2.5,2);
\draw[very thick](2.5,2)--(3,2.5);
\draw[very thick](3,2.5)--(3.5,2);
\draw[very thick,dotted](3.5,2)--(4.5,3);
\draw[thick,dashed] (2.25,0)--(2.25,1.75);
\draw[thick,dashed] (1,0)--(1,1.5);
\draw[thick,dashed] (3,0)--(3,2.5);
\draw[thick,dashed] (0,3)--(4.5,3);
\node[below] at(2.25,0){$\p2$};
\node[below] at(1,0){$d_1$};
\node[below] at(3,0){$d_2$};
\node[left]at(0,0.5){2};
\node[left]at(0,3){$q-2$};
\node[below]at(2,-1){(b) $\delta(P_1)=2$, $\pi(P_1)=0$, $d_2-d_1>3$};
\end{tikzpicture}
\mbox{}\\
\vskip 0.5cm
\begin{tikzpicture}[domain=0:11,x=1cm,y=1cm]
\draw[very thick,->] (-0.5,0) -- ++(5,0);
\draw[very thick,->] (0,-0.5) -- ++(0,3.5);
\node[right]at(4.5,0){$i$};
\node[above]at(0,3){$d^-_{P_1}(i)$};
\SOMMET{0,0.5}; \SOMMET{1,1.5}; \SOMMET{1.5,1}; \SOMMET{2,1.5}; \SOMMET{2.5,2}; \SOMMET{3,1.5}; \SOMMET{4,2.5};
\draw[very thick,dotted] (0,0.5)--(1,1.5);
\draw[very thick](1,1.5)--(1.5,1);
\draw[very thick](1.5,1)--(2,1.5);
\draw[very thick](2,1.5)--(2.5,2);
\draw[very thick](2.5,2)--(3,1.5);
\draw[very thick,dotted](3,1.5)--(4,2.5);
\draw[thick,dashed] (2,0)--(2,1.5);
\draw[thick,dashed] (1,0)--(1,1.5);
\draw[thick,dashed] (2.5,0)--(2.5,2);
\draw[thick,dashed] (0,2.5)--(4,2.5);
\node[below] at(2,0){$\p2$};
\node[below] at(1,0){$d_1$};
\node[below] at(2.5,0){$d_2$};
\node[left]at(0,0.5){2};
\node[left]at(0,2.5){$q-2$};
\node[below]at(1.75,-1){(c) $\delta(P_1)=2$, $\pi(P_1)=0$, $d_2-d_1=3$};
\end{tikzpicture}
\caption{Indegree paths for the proof of Theorem~\ref{th:S-atmost2}}
\label{fig:Natmost2a}
\end{center}
\end{figure}


\begin{enumerate}
\item $\delta(P_1)=2$ and $\pi(P_1)=0$.\\
In that case, the sequence $\IS(P_1)$ has no plateau and exactly two descent-vertices,
say $d_1$ and $d_2$.
Since $\IS(P_1)$ is centrally symmetric and $q$ is even, 
we necessarily have $d_2=q-d_1-1$ which implies in particular that
the difference $d_2-d_1$ is necessarily odd.

Suppose first that $d_2-d_1=1$ (see Figure~\ref{fig:Natmost2a}(a)).
In that case, we necessarily have $d_2={\p2}$.
Moreover, since $q>3$, the vertex 0 cannot be a descent-vertex (this would imply $\p2=1$)
and thus $1\notin S^-$.
The only set $V_d(P_1)$ containing more than two elements is
thus $V_{{\p2}}(P_1)=\{{\p2}-2,{\p2},{\p2}+2\}$.
But since the vertex ${\p2}-1$ is fixed by every automorphism of $P_1$
and vertices ${\p2}-2$ and ${\p2}$ disagree on $\p2-1$ (as $1\notin S^-$),
${\p2}-2$ and ${\p2}$ are also fixed so that $P_1$ is rigid
by Corollary~\ref{cor:half-is-fixed}.

If $d_2-d_1>3$ (see Figure~\ref{fig:Natmost2a}(b)), then every set $V_d(P_1)$ 
contains at most two elements and Lemma~\ref{lem:Vd-rigid} again allows us to conclude.

Suppose finally that $d_2-d_1=3$ (see Figure~\ref{fig:Natmost2a}(c)).
This case is similar to the case $d_2-d_1=1$: again, 
the only set $V_d(P_1)$ containing more than two elements is
$V_{{\p2}}(P_1)=\{{\p2}-2,{\p2},{\p2}+2\}$
while the vertex ${\p2}-1$ is fixed by every automorphism of $P_1$
and vertices ${\p2}-2$ and ${\p2}$ disagree on $\p2-1$.
Again, $P_1$ is rigid thanks to Corollary~\ref{cor:half-is-fixed}.

\item $\delta(P_1)=1$ and $\pi(P_1)=2$.\\
Since $\delta(P_1)=1$, we cannot have two consecutive descent-vertices in $P_1$.
Therefore, since $\IS(P_1)$ is centrally symmetric, if $i$ is a descent-vertex, then
$q-(i+1)$ is also a descent-vertex.
As $\delta(P_1)=1$, we must have $i=q-(i+1)$, in contradiction with our assumption
on the parity of~$q$.
Hence, this case cannot happen.

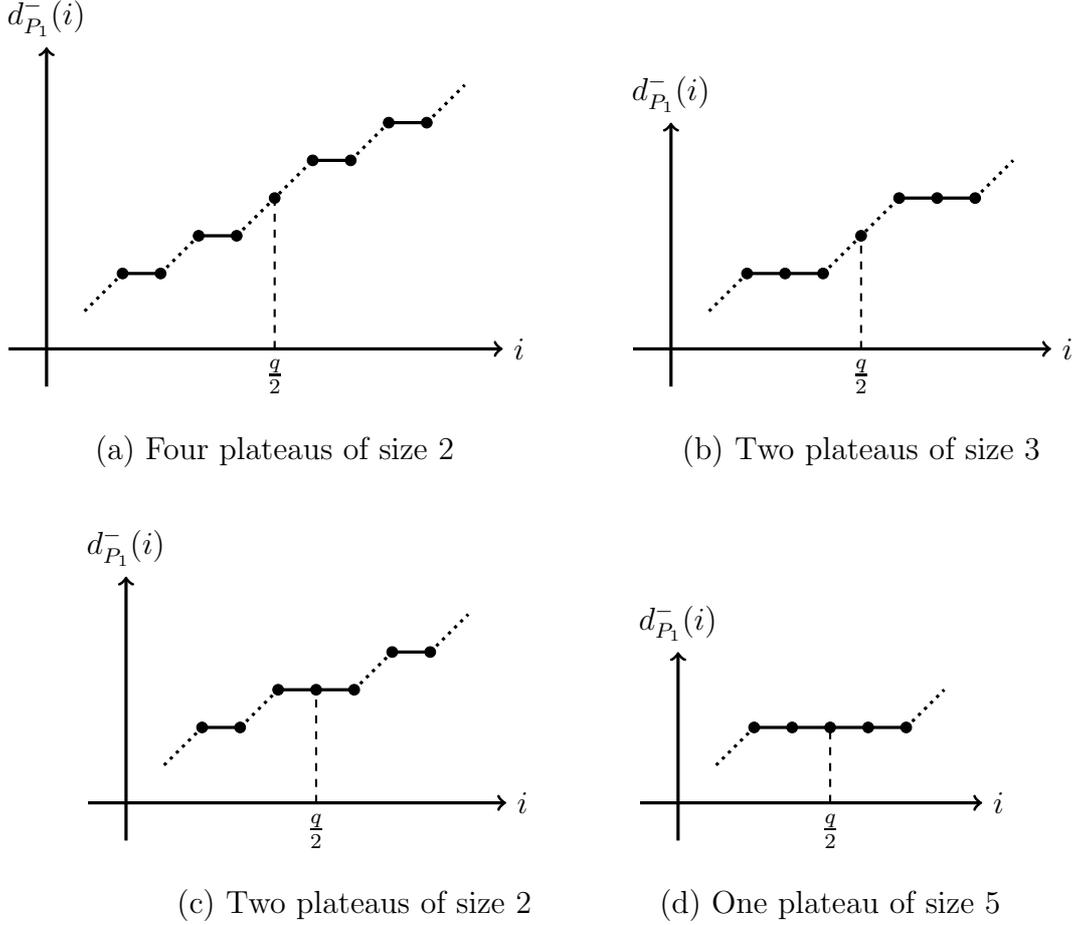
\begin{figure}
\begin{center}
\begin{tikzpicture}[domain=0:11,x=1cm,y=1cm]
\draw[very thick,->] (-0.5,0) -- ++(6.5,0);
\draw[very thick,->] (0,-0.5) -- ++(0,4.5);
\node[right]at(6,0){$i$};
\node[above]at(0,4){$d^-_{P_1}(i)$};
\SOMMET{1,1}; \SOMMET{1.5,1}; \SOMMET{2,1.5}; \SOMMET{2.5,1.5}; \SOMMET{3,2}; \SOMMET{3.5,2.5}; \SOMMET{4,2.5}; \SOMMET{4.5,3}; \SOMMET{5,3};
\draw[very thick,dotted] (0.5,0.5)--(1,1);
\draw[very thick](1,1)--(1.5,1);
\draw[very thick,dotted] (1.5,1)--(2,1.5);
\draw[very thick](2,1.5)--(2.5,1.5);
\draw[very thick,dotted] (2.5,1.5)--(3.5,2.5);
\draw[very thick](3.5,2.5)--(4,2.5);
\draw[very thick,dotted] (4,2.5)--(4.5,3);
\draw[very thick](4.5,3)--(5,3);
\draw[very thick,dotted](5,3)--(5.5,3.5);
\draw[thick,dashed] (3,0)--(3,2);
\node[below] at(3,0){$\p2$};
\node[below]at(3,-1){(a) Four plateaus of size 2};
\end{tikzpicture}
\hskip 1cm
\begin{tikzpicture}[domain=0:11,x=1cm,y=1cm]
\draw[very thick,->] (-0.5,0) -- ++(5.5,0);
\draw[very thick,->] (0,-0.5) -- ++(0,3.5);
\node[right]at(5,0){$i$};
\node[above]at(0,3){$d^-_{P_1}(i)$};
\SOMMET{1,1}; \SOMMET{1.5,1}; \SOMMET{2,1}; \SOMMET{2.5,1.5}; \SOMMET{3,2}; \SOMMET{3.5,2}; \SOMMET{4,2}; 
\draw[very thick,dotted] (0.5,0.5)--(1,1);
\draw[very thick](1,1)--(2,1);
\draw[very thick,dotted] (2,1)--(3,2);
\draw[very thick](3,2)--(4,2);
\draw[very thick,dotted] (4,2)--(4.5,2.5);
\draw[thick,dashed] (2.5,0)--(2.5,1.5);
\node[below] at(2.5,0){$\p2$};
\node[below]at(2.5,-1){(b) Two plateaus of size 3};
\end{tikzpicture}
\mbox{}\\
\vskip 0.5cm
\begin{tikzpicture}[domain=0:11,x=1cm,y=1cm]
\draw[very thick,->] (-0.5,0) -- ++(5.5,0);
\draw[very thick,->] (0,-0.5) -- ++(0,3.5);
\node[right]at(5,0){$i$};
\node[above]at(0,3){$d^-_{P_1}(i)$};
\SOMMET{1,1}; \SOMMET{1.5,1}; \SOMMET{2,1.5}; \SOMMET{2.5,1.5}; \SOMMET{3,1.5}; \SOMMET{3.5,2}; \SOMMET{4,2}; 
\draw[very thick,dotted] (0.5,0.5)--(1,1);
\draw[very thick](1,1)--(1.5,1);
\draw[very thick,dotted] (1.5,1)--(2,1.5);
\draw[very thick](2,1.5)--(3,1.5);
\draw[very thick,dotted] (3,1.5)--(3.5,2);
\draw[very thick](3.5,2)--(4,2);
\draw[very thick,dotted] (4,2)--(4.5,2.5);
\draw[thick,dashed] (2.5,0)--(2.5,1.5);
\node[below] at(2.5,0){$\p2$};
\node[below]at(3,-1){(c) Two plateaus of size 2};
\end{tikzpicture}
\hskip 1cm
\begin{tikzpicture}[domain=0:11,x=1cm,y=1cm]
\draw[very thick,->] (-0.5,0) -- ++(4.5,0);
\draw[very thick,->] (0,-0.5) -- ++(0,2.5);
\node[right]at(4,0){$i$};
\node[above]at(0,2){$d^-_{P_1}(i)$};
\SOMMET{1,1}; \SOMMET{1.5,1}; \SOMMET{2,1}; \SOMMET{2.5,1}; \SOMMET{3,1}; 
\draw[very thick,dotted] (0.5,0.5)--(1,1);
\draw[very thick](1,1)--(3,1);
\draw[very thick,dotted] (3,1)--(3.5,1.5);
\draw[thick,dashed] (2,0)--(2,1);
\node[below] at(2,0){$\p2$};
\node[below]at(2,-1){(d) One plateau of size 5};
\end{tikzpicture}
\caption{Indegree paths for the proof of Theorem~\ref{th:S-atmost2} (Cont.)}
\label{fig:Natmost2b}
\end{center}
\end{figure}


\item $\delta(P_1)=0$ and $\pi(P_1)=4$.\\
Since $q$ is even and $\IS(P_1)$ is centrally symmetric, we have
four possibilities: (a) $P_1$ contains four plateaus of size 2,
(b) $P_1$ contains two plateaus of size 3,
(c) $P_1$ contains one plateau of size 3, centered at ${\p2}$, and two plateaus of size 2,
or (d) $P_1$ contains a single plateau, of size 5 and centered at ${\p2}$,
see Figure~\ref{fig:Natmost2b}(a), (b), (c) or (d), respectively.

\begin{enumerate}
\item If $P_1$ contains four plateaus of size 2, 
then $P_1$ satisfies the condition of Lemma~\ref{lem:Vd-rigid} and we are done.

\item If $P_1$ contains two plateaus of size 3, then the plateau-vertices
are necessarily of the form $i$, $i+1$, $q-i-2$ and $q-i-1$, with $i\le{\p2}-3$,
so that $S^-=\{a_1,a_2\}$, with $a_1\in\{i+1,q-i\}$, $a_2\in\{i+2,q-i-1\}$, and $d^-_{P_1}(i)=2+i$.
We then have $V_{2+i}(P_1)=\{i,i+1,i+2\}$ and all other sets $V_d(P_1)$, $2\le d\le{\p2}$, $d\neq 2+i$,
are singletons.

If $i>0$ and $i+1\in S^-$, then, since $i\notin S^-$, vertices $i$ and $i+1$ disagree on 0.
Since vertex 0 is fixed by every automorphism of $P_1$, $i$ and $i+1$ cannot belong to the same
orbit of any automorphism and thus, by Corollary~\ref{cor:half-is-fixed}, $P_1$ is rigid.
If $i>0$ and $q-i\in S^-$, 
we get the same conclusion since, in that case, 
either $i+2\in S^-$, which implies that $i$ and $i+2$ disagree on 0,
or $q-i-1\in S^-$, which implies, since $q-i-2\notin S^-$, that $i$ and $i+2$ disagree on $q$
(and $q$ is fixed by every automorphism of $P_1$).
Suppose finally that $i=0$.
Since $q$ is even, the vertex $3$ is not a plateau-vertex and thus is fixed
by every automorphism of $P_1$.
If $0$, $1$ and $2$ disagree on $3$, these three vertices
cannot belong to the same orbit of any automorphism and $P_1$ is rigid thanks to Corollary~\ref{cor:half-is-fixed}.
Otherwise, since 3 cannot be an in-neighbour of 0, 1 and 2 (this would imply $|S^-|\ge 3$),
3 is an out-neighbour of 0, 1 and 2, so that $1,2\notin S^-$, which implies
that $P_1[\{0,1,2\}]$ is transitive, and $P_1$ is rigid by Lemma~\ref{lem:Vd-rigid}.

\item If $P_1$ contains one plateau of size 3 and two plateaus of size 2,
then the plateau-vertices
are necessarily of the form $i$, ${\p2}-1$, ${\p2}$ and $q-i-1$, with $i\le{\p2}-3$,
so that $S^-=\{a_1,a_2\}$, with $a_1\in\{i+1,q-i\}$, $a_2\in\{{\p2},{\p2}+1\}$,
and $d^-_{P_1}({\p2})={\p2}$.
Therefore, the only set $V_d(P_1)$ containing more than two elements is
$V_{{\p2}}(P_1)=\{{\p2}-1,{\p2},{\p2}+1\}$.
Since $S^-$ contains either ${\p2}$ or ${\p2}+1$,
vertices ${\p2}$ and ${\p2}+1$ disagree on 0.
Since  0 is fixed by every automorphism of $P_1$, ${\p2}-1$, ${\p2}$ and ${\p2}+1$ are also fixed and thus, 
by Corollary~\ref{cor:half-is-fixed}, $P_1$ is rigid.

\item If $P_1$ contains a single plateau of size 5,
then the plateau-vertices
are necessarily ${\p2}-2$, ${\p2}-1$, $\dots$, ${\p2}+1$,
so that $S^-=\{a_1,a_2\}$, with $a_1\in\{{\p2}-1,{\p2}+2\}$, 
$a_2\in\{{\p2},{\p2}+1\}$,
and $d^-_{P_1}({\p2})={\p2}$.
Moreover, if $q>4$, then every set $V_d(P_1)$ except $V_{{\p2}}(P_1)$ is a singleton.

Suppose first that $q>4$, so that $1\notin S^-$ (otherwise, this would imply $1=\p2-1$).

If $\{2,3\}\cap S^-=\emptyset$ (which happens in particular if $q\ge 10$), 
the sub-tournament induced by $V_{{\p2}}(P_1)$
is either transitive or almost transitive, and thus rigid, so that $P_1$ is
rigid by Lemma~\ref{lem:Vd-rigid} (this also follows from Theorem~\ref{th:minN}).

Suppose now that $q\in\{6,8\}$, so that vertices 0 and $q$ are fixed by
every automorphism of $P_1$.
If all 
vertices belonging to the plateau are fixed by every automorphism of $P_1$,
then $P_1$ is rigid by Lemma~\ref{lem:Vd-rigid}.
Otherwise, some 
of these vertices may form an orbit, say $O$, of some automorphism, of size 3 or~5.
Considering the structure of the set $S^-$, we get that the vertex 0 has
three out-neighbours and two in-neighbours in 
$\{{\p2}-2,{\p2}-1,{\p2},{\p2}+1,{\p2}+2\}$.
Therefore, the size of the orbit $O$ must be 3. 
Moreover, since $0({\p2}-x)$ is an arc if and only if $({\p2}+x)p$
if an arc for every $x\in\{0,1,2\}$, we get that this orbit, if it exists,
is necessarily $O=\{{\p2}-2,{\p2},{\p2}+2\}$,
and thus $S^-=\{{\p2}-1,{\p2}+1\}$.
Moreover, $O$ must induce a 3-cycle in $P_1$, which implies
either $2\in S^-$ and $4\notin S^-$, or $2\notin S^-$ and $4\in S^-$.
This forces $S^-=\{4,6\}$ and thus $\{2,3\}\cap S^-=\emptyset$ 
for which we have seen before that $P_1$ is rigid.

Suppose finally that $q=4$.
Without loss of generality, we can suppose that $1\notin S^-$
(otherwise, we consider $P^c$ instead of~$P$ and we have $|(S^-)^c|=|S^-|$,
so that either $S^-=\{2,4\}$ or $S^-=\{3,4\}$ and,
in both cases, $P_1$ is a cyclic tournament).
Due to our initial assumption on the parity of~$q$, this case occurs
either if $T=T(9;S^-)$, in which case $P(4;S^-)=P_1$, or
$T=T(11;S^-)$, in which case $P(4;S^-)=P_2$, so that $P_1=P(5;S^-)$.

In the former case,
the pseudo-cyclic tournament $P_2$ is then either $P(3;\{2\})$ or $P(3;\{3\})$,
respectively. In both cases, $P_2$ is rigid since $|S^-|=1$.

In the latter case, we will prove that $P=P(5;S^-)$ is rigid. We consider two subcases,
according to the set $S^-$.
\begin{enumerate}
\item $S^-=\{2,4\}$.\\
In that case, we have $V_2(P)=\{0,2,4\}$ and $V_3(P)=\{1,3,5\}$.
Since both sub-tournaments $P[V_2(P)]$ and $P[V_3(P)]$ are transitive,
$P$ is rigid by Lemma~\ref{lem:Vd-rigid}.

\item $S^-=\{3,4\}$.\\
In that case, we have $V_2(P)=\{0,3,4\}$ and $V_3(P)=\{1,2,5\}$.
Again, both sub-tournaments $P[V_2(P)]$ and $P[V_3(P)]$ are transitive,
$P$ is rigid by Lemma~\ref{lem:Vd-rigid}.
\end{enumerate}

\end{enumerate}

\end{enumerate}

This completes the proof.
\end{proof}

For proving Theorem~\ref{th:S-atmost2}, we showed that whenever $|S^-|\le 2$,
at least one of the two sub-tournaments $T_{0,q}$ or $T_{q+1,2q}$ of the cyclic
tournament $T=T(2q+1;S)$ is rigid.
It should be noticed that this property does not hold in general when $|S^-|\ge 3$,
as shown by the cyclic tournament $T=T(13;\{2,5,6\})$ given in Example~\ref{ex:T13}.

Since any tournament $T$ is isomorphic to its converse, Theorem~\ref{th:S-atmost2}
gives the following.

\begin{corollary}
If $T=T(2q+1;S)$ is a cyclic tournament, such that $|S^-|\ge q-1$,
then $T$ satisfies the Albertson-Collins Conjecture.
\label{cor:atmostq-1}
\end{corollary}

\subsection{The set $S^-$ forms an interval}
\label{subsec:interval}

Observe first the following easy result.

\begin{theorem}
For every $q\ge 1$,
if $S=\{1,\dots,q\}$ or $S=\{1,\dots,q-1,q+1\}$, then the cyclic tournament
$T(2q+1;S)$ satisfies the Albertson-Collins Conjecture.
\label{th:transitive-and-almost}
\end{theorem}

\begin{proof}
This directly follows from Theorem~\ref{th:S-atmost2} since we have either $|S^-|=0$ or $|S^-|=1$.
\end{proof}

In fact, we can prove that whenever the set $S^-$ of negative connectors forms an interval of integers,
the corresponding cyclic tournament $T(2q+1;S)$ satisfies the Albertson-Collins Conjecture.

We first prove two lemmas.
The first one says that for every pseudo-cyclic tournament $P$,
 if vertices $0,\dots,i-1$ are fixed by every automorphism of $P$,
$i$ is the first vertex of a plateau whose corresponding negative connectors
 form an interval and every vertex outside the plateau having the same indegree as the vertices
 of the plateau is fixed by every automorphism of $P$, then all vertices of the plateau are also
 fixed by every automorphism of $P$.

\begin{lemma}
Let $P=P(q;N)$ be a pseudo-cyclic tournament, and $i$ a vertex of $P$ with $i<\p2$,
satisfying all the following conditions.
\begin{itemize}
\item[{\rm (i)}] Every vertex $j$, $0\le j<i$, is fixed by every automorphism of $P$.
\item[{\rm (ii)}] $i-1$ is an ascent vertex.
\item[{\rm (iii)}] $(i,i+1,\dots,i+k-1)$ is a plateau of size $k\ge 2$ such that one of
the two following conditions holds.
  \begin{itemize}
  \item[{\rm (a)}] $i+k-1<{\p2}$ and either $[i+1,i+k-1]\subseteq N$, or $[q-i-k+2,q-i]\subseteq N$.
  \item[{\rm (b)}] $i+k-1>{\p2}$ and either $[i+1,\p2]\subseteq N$, or $[\p2+1,q-i]\subseteq N$.
  \end{itemize} 
\item[{\rm (iv)}] $|V_d(P)|\ge k$, for $d=d^-_P(i)$, and, if $|V_d(P)| > k$, then every vertex in
  $V_d(P)\setminus\{i,i+1,\dots,i+k-1\}$ is fixed by every automorphism of $P$.
\end{itemize}
Then, every vertex in $\{i,i+1,\dots,i+k-1\}$ is fixed by every automorphism of $P$.
\label{lem:rigid-plateau}
\end{lemma}

\begin{proof}
Suppose that $P$ and $i$ satisfy the conditions of the lemma.
According to condition (iv), it is enough to prove that no two vertices in $\{i,i+1,\dots,i+k-1\}$
can belong to the same orbit of any automorphism of $P$, 
since this implies that every such vertex is fixed by every automorphism of $P$.
Observe that $i\notin N$ since $i-1$ is an ascent-vertex.

We consider two cases, according to the position of the plateau.
\begin{enumerate}
\item $i+k-1<{\p2}$.\\
We consider two subcases, depending on the corresponding interval of negative connectors.
\begin{enumerate}
\item $[i+1,i+k-1]\subseteq N$.\label{item:utile1}\\
In that case, $P$ contains arcs $0i$, $(i+1)0$, $\dots$, $(i+k-1)0$,
so that 
 $i$ and all other vertices of the plateau disagree on 0,
and thus, by Observation~\ref{obs:agree-fixpoint}, $i$ and any other vertex of the plateau
cannot belong to the same orbit 
of any automorphism of~$P$.
Similarly, for every $j=i+1,\dots,i+k-1$ (in that order),
$j$ and 
vertices $j+1,\dots,i+k-1$ disagree on $j-i$ 
(which is fixed by every automorphism of~$P$),
so that, by Observation~\ref{obs:agree-fixpoint}, $j$  and any other vertex of the plateau
cannot belong to the same orbit of any automorphism of~$P$.
Therefore, all vertices of the plateau are fixed by every automorphism of~$P$.

\item $[q-i-k+2,q-i]\subseteq N$.\label{item:utile2}\\
If $N$ contains no negative connector $a<i$, then the sub-tournament
induced by the plateau is transitive
(since we then have $\min(N)>i+k-1$), so that all  vertices 
of the plateau are fixed by every automorphism of~$P$.

Otherwise, let $a$ be the largest negative connector such that $a<i$.
We then have $[a+1,i+k-1]\cap N=\emptyset$, so that 
$i$  and all other vertices of the plateau disagree on $i-a$,
and thus, by Observation~\ref{obs:agree-fixpoint}, 
$i$ and any other vertex of the plateau cannot belong to the same orbit of any automorphism of~$P$.
Similarly, for every $j=i+1,\dots,i+k-1$ (in that order),
$j$~and 
vertices $j+1,\dots,i+k-1$ disagree on $j-a$ (which is fixed by every automorphism of~$P$),
so that, by Observation~\ref{obs:agree-fixpoint}, 
$j$~and any other vertex of the plateau cannot belong to the same orbit of any automorphism of~$P$.
Therefore, all vertices of the plateau are fixed by every automorphism of~$P$.
\end{enumerate}

\item $i+k-1>{\p2}$.\\
In that case, $q$ is necessarily even.
Since $\IS(P)$ is centrally symmetric, we have 
$(i,i+1,\dots,i+k-1)=({\p2}-r,\dots,{\p2}+r)$, for some $r\ge 1$.

If $\left[{\p2}-r+1,{\p2}\right]\subseteq N$, then,
using the same proof as in Case~\ref{item:utile1} above, we get that no two vertices in 
$\{{\p2}-r,\dots,{\p2}\}$ can belong
to the same orbit of any automorphism of~$P$.
Symmetrically, starting with vertex $q$ instead of vertex 0
(recall that $q-j$, $j\ge 0$, is fixed whenever $j$ is fixed, 
by Proposition~\ref{prop:fixed-by-symmetry}), we can prove similarly that no two vertices in 
$\{{\p2},\dots,{\p2}+r\}$ can belong
to the same orbit of any automorphism of~$P$.
Therefore, no vertex of the plateau can belong to an orbit of size at least 3
of any automorphism of~$P$, which implies
that every vertex in $\{{\p2}-r,\dots,{\p2}+r\}$
is fixed by every automorphism of~$P$ and we are done.

If $\left[{\p2}+1,{\p2}+r\right]\subseteq N$, then we proceed as in Case~\ref{item:utile2} above: 
if the vertices of the plateau do not induce a transitive tournament, then
(i) using the largest negative connector 
$a<{\p2}-r$, we get that no two vertices in 
$\{{\p2}-r,\dots,{\p2}\}$ can belong
to the same orbit of any automorphism of~$P$,
and (ii) symmetrically, using the smallest negative connector 
$a'>{\p2}+r$, we get that
no two vertices in 
$\{{\p2},\dots,{\p2}+r\}$ can belong
to the same orbit of any automorphism of~$P$
and thus, again, every vertex in $\{{\p2}-r,\dots,{\p2}+r\}$
is fixed by every automorphism of~$P$.
\end{enumerate}
This concludes the proof.
\end{proof}

The second lemma says that every pseudo-cyclic tournament containing either ascent-vertices and no 
descent-vertices, or descent-vertices and no ascent-vertices,
is rigid whenever every plateau, if any, is produced by an interval of negative connectors.

\begin{figure}
\begin{center}
\begin{tikzpicture}[domain=0:11,x=1cm,y=1cm]
\draw[very thick,->] (-0.5,0) -- ++(4,0);
\draw[very thick,->] (0,-0.5) -- ++(0,3.5);
\node[right]at(3.5,0){$i$};
\node[above]at(0,3){$d^-_P(i)$};
\SOMMET{0,2}; \SOMMET{0.5,2.5}; \SOMMET{1.5,1.5}; \SOMMET{2.5,0.5}; \SOMMET{3,1}; 
\draw[very thick,dotted] (0,2)--(0.5,2.5);
\draw[very thick,dotted] (0.5,2.5)--(2.5,0.5);
\draw[very thick,dotted] (2.5,0.5)--(3,1);
\draw[thick,dashed] (1.5,0)--(1.5,1.5);
\draw[thick,dashed] (0.5,0)--(0.5,2.5);
\draw[thick,dashed] (2.5,0)--(2.5,0.5);
\draw[thick,dashed] (0,1)--(3,1);
\node[below] at(1.5,0){$\p2$};
\node[below] at(0.5,0){$a-1$};
\node[below] at(2.5,0){$b$};
\node[left]at(0,2){$b-a+1$};
\node[left]at(0,1){$q-b+a-1$};
\node[below]at(0.64,-1){(a) $a+b=q+1$};
\end{tikzpicture}
\hskip 1cm
\begin{tikzpicture}[domain=0:11,x=1cm,y=1cm]
\draw[very thick,->] (-0.5,0) -- ++(5,0);
\draw[very thick,->] (0,-0.5) -- ++(0,3.5);
\node[right]at(4.5,0){$i$};
\node[above]at(0,3){$d^-_P(i)$};
\SOMMET{0,2.5}; \SOMMET{1,2.5}; \SOMMET{2,1.5}; \SOMMET{3,0.5}; \SOMMET{4,0.5}; 
\draw[very thick,dotted] (0,2.5)--(1,2.5);
\draw[very thick,dotted] (1,2.5)--(3,0.5);
\draw[very thick,dotted] (3,0.5)--(4,0.5);
\draw[thick,dashed] (2,0)--(2,1.5);
\draw[thick,dashed] (1,0)--(1,2.5);
\draw[thick,dashed] (3,0)--(3,0.5);
\draw[thick,dashed] (0,0.5)--(3,0.5);
\node[below] at(2,0){$\p2$};
\node[below] at(1,0){$q-b$};
\node[below] at(3,0){$b$};
\node[left]at(0,2.5){$b$};
\node[left]at(0,0.5){$q-b$};
\node[below]at(1.7,-1){(b) $a=1$};
\end{tikzpicture}
\mbox{}\\
\vskip 0.5cm
\begin{tikzpicture}[domain=0:11,x=1cm,y=1cm]
\draw[very thick,->] (-0.5,0) -- ++(6,0);
\draw[very thick,->] (0,-0.5) -- ++(0,3.5);
\node[right]at(5.5,0){$i$};
\node[above]at(0,3){$d^-_P(i)$};
\SOMMET{0,2}; \SOMMET{0.5,2.5}; \SOMMET{1.5,2.5}; \SOMMET{2.5,1.5}; \SOMMET{3.5,0.5}; \SOMMET{4.5,0.5}; \SOMMET{5,1};
\draw[very thick,dotted] (0,2)--(0.5,2.5);
\draw[very thick,dotted] (0.5,2.5)--(1.5,2.5);
\draw[very thick,dotted] (1.5,2.5)--(3.5,0.5);
\draw[very thick,dotted] (3.5,0.5)--(4.5,0.5);
\draw[very thick,dotted] (4.5,0.5)--(5,1);
\draw[thick,dashed] (2.5,0)--(2.5,1.5);
\draw[thick,dashed] (0,1)--(5,1);
\node[below] at(2.5,0){$\p2$};
\node[left]at(0,2){$b-a+1$};
\node[left]at(0,1){$q-b+a-1$};
\node[below]at(1.7,-1){(c) $1<a$, $b<q$ and $a+b\neq q+1$};
\end{tikzpicture}
\caption{Indegree paths for the proof of Theorem~\ref{th:interval}, case $b-a+1>\p2$}
\label{fig:Nab-1}
\end{center}
\end{figure}
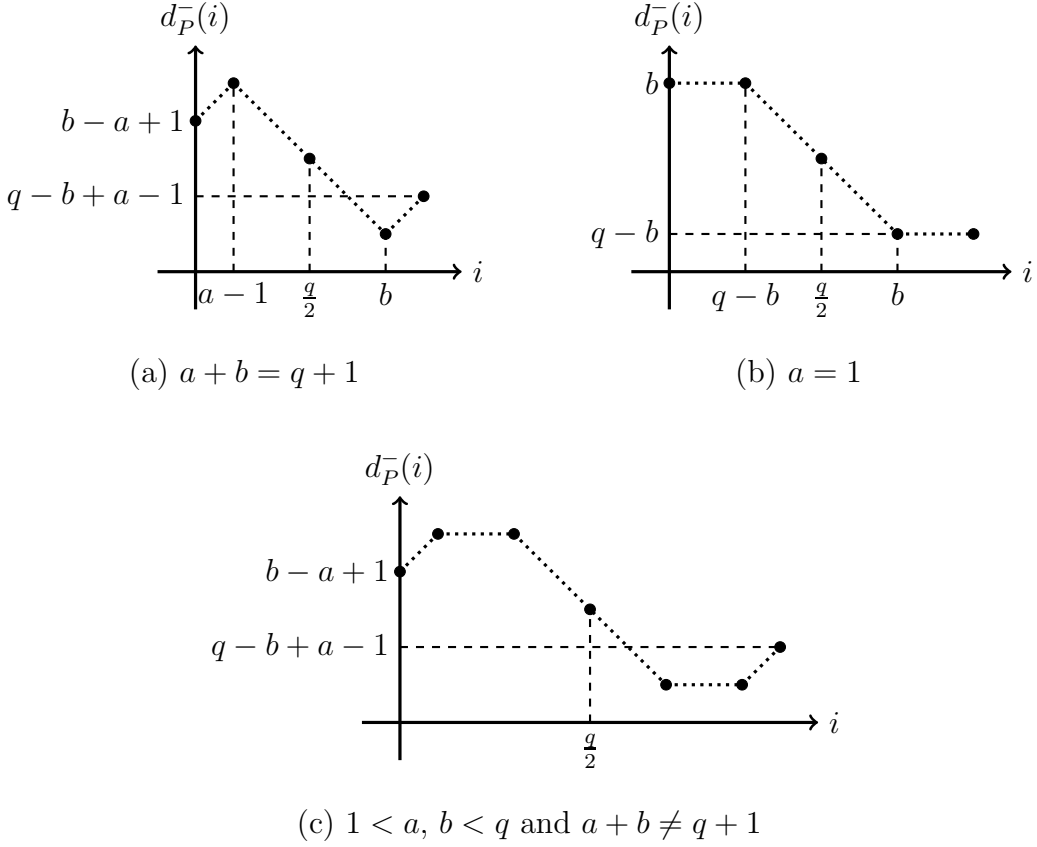


\begin{lemma}
If $P=P(q;N)$ is a pseudo-cyclic tournament such that 
\begin{itemize}
\item[{\rm (i)}] $\alpha(P)+\delta(P)>0$,
\item[{\rm (ii)}] $\alpha(P)\delta(P)=0$, 
\item[{\rm (iii)}] for every plateau $(i,i+1,\dots,i+k-1)$ of size $k\ge 3$, with $i+k-1<{\p2}$,
either $[i+1,i+k-1]\subseteq N$, or $[q-i-k+2,q-i]\subseteq N$, and
\item[{\rm (iv)}] for every plateau $({\p2}-q,\dots,{\p2}+q)$, $q$ even, $q\ge 1$,
of size $2q+1$,
either $\left[{\p2}-q+1,{\p2}\right]\subseteq N$, 
or $\left[{\p2}+1,{\p2}+q\right]\subseteq N$,
\end{itemize}
then $P$ is rigid.
\label{lem:ascent-plateau}
\end{lemma}

\begin{proof}
If $\pi(P)=0$, then $P$ is a transitive tournament and is thus rigid.

Assume now that $\pi(P)>0$ and
suppose first that $P$ has no descent-vertices, that is, $\delta(P)=0$ and $\alpha(P)>0$.
We will prove by induction on $i$ that every vertex $i\le\left\lfloor{\p2}\right\rfloor$
is fixed by every automorphism of~$P$, 
so that the result follows by Corollary~\ref{cor:half-is-fixed}.

If 0 is an ascent-vertex or the first vertex of a plateau of size 2, 
then 0 is fixed by every automorphism of~$P$.
If 0 is the first vertex of a plateau $(0,\dots,k-1)$ of size $k\ge 3$,
then the sub-tournament induced by this
plateau is transitive since we have either $[1,k-1]\in N$ or $[q-k+2,q]\in N$. 
Therefore, all vertices $0,\dots,k-1$  are fixed by every automorphism of~$P$.

Suppose now that all vertices $0,\dots,i-1<\left\lfloor{\p2}\right\rfloor$
are fixed by every automorphism of~$P$.
If $i$ is an ascent-vertex, then $i$ is fixed by every automorphism of~$P$,
otherwise $i$ is the first vertex of a plateau and 
the result follows from Lemma~\ref{lem:rigid-plateau}.

The case $\alpha(P)=0$ and $\delta(P)>0$ follows from Observation~\ref{obs:converse-pseudo}, 
considering $P^c$ instead of~$P$ (we then have $\delta(P^c)=0$ and $\alpha(P^c)>0$).
\end{proof}

Note that conditions (iii) and (iv) in Lemma~\ref{lem:ascent-plateau}
are both necessary. For instance, the pseudo-cyclic tournament $P=P(5;\{2,5\})$
has only ascent-vertices and plateau-vertices (we have $\IS(P)=(2,2,2,3,3,3)$)
while $(0,1,2)(3,4,5)$ is an automorphism of~$P$ so that $P$ is not rigid,
and thus condition (iii) is necessary.
On the other hand, the pseudo-cyclic tournament $P=P(8;\{2,3,5\})$
has only ascent-vertices and plateau-vertices, but a single central plateau
 (we have $\IS(P)=(3,4,\dots,4,5)$)
while $(1,4,7)$ is an automorphism of~$P$ so that $P$ is not rigid,
and thus condition (iv) is necessary.

We are now able to prove the following result.

\begin{theorem}
\label{th:interval}
If $T=T(2q+1;S)$ is a cyclic tournament, such that $S^-=[a,b]$, $1\le a\le b\le q$,
then $T$ satisfies the Albertson-Collins Conjecture.
\end{theorem}

\begin{proof}
We will prove that $P=P(q;[a,b])$ is rigid whenever $q$ is even.
Therefore, considering either $T_{0,q}$ or $T_{q+1,2q}$ depending
on the parity of $q$, the result will follow by Proposition~\ref{prop:T1orT2isrigid}.

If $b\le a+1$, the result follows from Theorem~\ref{th:S-atmost2}, so that we can assume
$b\ge a+2$.
Suppose now that ${\p2}\ge b$ or ${\p2}<a$.
If $[a,b]\neq\left[1,{\p2}\right]$ and $[a,b]\neq\left[{\p2}+1,q\right]$,
then $T_{0,q}$ is rigid by Lemma~\ref{lem:ascent-plateau}.
Otherwise, $T_{q+1,2q}$ is rigid by Lemma~\ref{lem:ascent-plateau}
(since $\IS(T_{q+1,2q})=({\p2},\dots,{\p2},{\p2}-1,\dots,{\p2}-1)$ 
if $[a,b]=\left[1,{\p2}\right]$,
and $\IS(T_{q+1,2q})=({\p2}-1,\dots,{\p2}-1,{\p2},\dots,{\p2})$ 
if $[a,b]=\left[{\p2}+1,q\right]$).

From now on, we thus assume that $a\le {\p2}< b$.
We consider three subcases, depending on the size of $[a,b]$.

\begin{enumerate}
\item $b-a+1>{\p2}$.\\
In that case, $d^-_P(0)=b-a+1>q-(b-a+1)=d^-_P(q)$.
We consider three subcases, corresponding to the three possible forms
of the indegree path of~$P$.
\begin{enumerate}
\item $a+b=q+1$.\\
In that case, no vertex in $P$ can be a plateau-vertex, and thus
the indegree path of~$P$ contains an ascent, a descent and an ascent
(see Figure~\ref{fig:Nab-1}(a)). 
Therefore, each set $V_d(P)$ has cardinality at most 2 and $P$ is rigid by Lemma~\ref{lem:Vd-rigid}.

\item $a=1$ or $b=q$.\\
In that case, vertices $0,\dots,q-b-1$ (if $a=1$)
or $0,\dots,a-1$ (if $q=b$)
are plateau-vertices, and thus
the indegree path of~$P$ contains a plateau of size $q-b+1$ or $a+1$, a descent and a plateau
of size $q-b+1$ or $a+1$ (see Figure~\ref{fig:Nab-1}(b) for the case $a=1$). 
In both cases, the sub-tournaments induced by these plateaus are transitive, so
that every vertex of~$P$ is fixed by every automorphism of~$P$ by Observation~\ref{obs:orbit-regular},
and $P$ is rigid
by Corollary~\ref{cor:half-is-fixed}.

\item $1<a$, $b<p$ and $a+b\neq q+1$.\label{item:utile}\\
In that case, the indegree path of~$P$ contains an ascent, a plateau, a descent, a plateau
and an ascent (see Figure~\ref{fig:Nab-1}(c)). 
Every vertex not belonging to a plateau belongs to a set $V_d(P)$
with cardinality~2 and is thus fixed by every automorphism of~$P$.
By Lemma~\ref{lem:rigid-plateau}, every vertex belonging to a plateau is also
fixed by every automorphism of~$P$.
Using now Corollary~\ref{cor:half-is-fixed}, we get that $P$ is rigid.

\end{enumerate}

\begin{figure}
\begin{center}
\begin{tikzpicture}[domain=0:11,x=1cm,y=1cm]
\draw[very thick,->] (-0.5,0) -- ++(5,0);
\draw[very thick,->] (0,-0.5) -- ++(0,3.5);
\node[right]at(4.5,0){$i$};
\node[above]at(0,3){$d^-_P(i)$};
\SOMMET{0,1.5}; \SOMMET{1,2.5}; \SOMMET{2,1.5}; \SOMMET{3,0.5}; \SOMMET{4,1.5}; 
\draw[very thick,dotted] (0,1.5)--(1,2.5);
\draw[very thick,dotted] (1,2.5)--(3,0.5);
\draw[very thick,dotted] (3,0.5)--(4,1.5);
\draw[thick,dashed] (1,0)--(1,2.5);
\draw[thick,dashed] (2,0)--(2,1.5);
\draw[thick,dashed] (3,0)--(3,0.5);
\draw[thick,dashed] (0,1.5)--(4,1.5);
\node[below] at(2,0){$\p2$};
\node[below] at(1,0){$a-1$};
\node[below] at(3,0){$b$};
\node[left]at(0,1.5){$b-a+1$};
\node[below]at(1.5,-1){(a) $a+b=q+1$};
\end{tikzpicture}
\mbox{}\\
\vskip 0.5cm
\begin{tikzpicture}[domain=0:11,x=1cm,y=1cm]
\draw[very thick,->] (-0.5,0) -- ++(7,0);
\draw[very thick,->] (0,-0.5) -- ++(0,3.5);
\node[right]at(6.5,0){$i$};
\node[above]at(0,3){$d^-_P(i)$};
\SOMMET{0,1.5}; \SOMMET{1,2.5}; \SOMMET{2,2.5}; \SOMMET{3,1.5}; \SOMMET{4,0.5}; \SOMMET{5,0.5}; \SOMMET{6,1.5};
\draw[very thick,dotted] (0,1.5)--(1,2.5);
\draw[very thick,dotted] (1,2.5)--(2,2.5);
\draw[very thick,dotted] (2,2.5)--(4,0.5);
\draw[very thick,dotted] (4,0.5)--(5,0.5);
\draw[very thick,dotted] (5,0.5)--(6,1.5);
\draw[thick,dashed] (3,0)--(3,1.5);
\draw[thick,dashed] (1,0)--(1,2.5);
\draw[thick,dashed] (5,0)--(5,0.5);
\draw[thick,dashed] (0,1.5)--(6,1.5);
\node[below] at(3,0){$\p2$};
\node[below] at(1,0){$a-1$};
\node[below] at(5,0){$q-a+1$};
\node[left]at(0,1.5){$b-a+1$};
\node[below]at(2.5,-1){(b) $a+b\neq q+1$};
\end{tikzpicture}
\caption{Indegree paths for the proof of Theorem~\ref{th:interval}, case $b-a+1=\p2$}
\label{fig:Nab-2}
\end{center}
\end{figure}


\item $b-a+1={\p2}$.\\
Note first that since $a\le {\p2}< b$, we necessarily have $a>1$ and $b<p$.
The indegree path of~$P$
contains an ascent, possibly a plateau, a descent, possibly a plateau and an ascent.
More precisely, $\IS(P)$ contains a plateau if and only if $a+b\neq q+1$
(see Figure~\ref{fig:Nab-2}(a) and (b)).
Moreover, since $d^-_P(0)=d^-_P(q)=\p2$, we get $V_\p2(P)=\{0,\p2,q\}$
(and the sub-tournament induced by $V_\p2(P)$ is a 3-cycle).
Only two other sets $V_d(P)$ may have cardinality at least~3, namely those
two sets corresponding to the two plateaus.

We consider three subcases, depending on the values of $a$ and $b$.

\begin{enumerate}
\item $a>2$ and $b<q-1$.\\
We first prove that every vertex in $V_\p2(P)=\{0,\p2,q\}$ is fixed by every automorphism of~$P$.
Since the vertex 1 is an ascent-vertex with $d^-_P(1)=\p2+1$ and $|V_{\p2+1}|=2$, it is fixed by every automorphism of~$P$.
Since $q-1\notin[a,b]$, vertices 0 and $q$ disagree on~1,
so that, by Observation~\ref{obs:agree-fixpoint}, 0 and $q$ cannot belong to the same orbit 
of any automorphism of~$P$ and we are done.

Moreover, by Lemma~\ref{lem:rigid-plateau}, all vertices of the first plateau 
(whenever $a+b\neq q+1$)
are fixed by every automorphism of~$P$.
Hence, all vertices $i\le\p2$ are fixed by every automorphism of~$P$ and
the result follows by Corollary~\ref{cor:half-is-fixed}.

\item $b=q-1$.\\
In that case, we thus have $a=\p2$.
The first plateau contains vertices $1$ to $\p2-1$
and thus induces a transitive tournament, so that all its vertices are 
fixed by every automorphism of~$P$.
Now, since vertices $\p2$ and $q$ disagree on $\p2-1$,
the three vertices $\{0,\p2,q\}$ are all fixed by every automorphism of~$P$
and the result follows by Corollary~\ref{cor:half-is-fixed}.

\item $a=2$.\\
In that case, we thus have $b={\p2}+1$.
The first plateau contains vertices $1$ to $\p2-1$
and is isomorphic to the pseudo-cyclic tournament $PC=P(\p2-2;[2,\p2-2])$.
Since $PC^c$ is rigid by Theorem~\ref{th:S-atmost2}, $PC$ is rigid.
Now, since vertices $0$ and $q$ disagree on $1$,
the three vertices $\{0,\p2,q\}$ are all fixed by every automorphism of~$P$
and, again, the result follows by Corollary~\ref{cor:half-is-fixed}.

\end{enumerate}

\begin{figure}
\begin{center}
\begin{tikzpicture}[domain=0:11,x=1cm,y=1cm]
\draw[very thick,->] (-0.5,0) -- ++(5,0);
\draw[very thick,->] (0,-0.5) -- ++(0,3.5);
\node[right]at(4.5,0){$i$};
\node[above]at(0,3){$d^-_P(i)$};
\SOMMET{0,0.5}; \SOMMET{1.5,2}; \SOMMET{2,1.5}; \SOMMET{2.5,1}; \SOMMET{4,2.5}; 
\draw[very thick,dotted] (0,0.5)--(1.5,2);
\draw[very thick,dotted] (1.5,2)--(2.5,1);
\draw[very thick,dotted] (2.5,1)--(4,2.5);
\draw[thick,dashed] (1.5,0)--(1.5,2);
\draw[thick,dashed] (2.5,0)--(2.5,1);
\draw[thick,dashed] (0,2)--(1.5,2);
\draw[thick,dashed] (0,1)--(2.5,1);
\node[below] at(1.5,0){$a-1$};
\node[below] at(2.5,0){$b$};
\node[left]at(0,0.5){$b-a+1$};
\node[left]at(0,1){$q-b$};
\node[left]at(0,2){$b$};
\node[below]at(1.5,-1){(a) $a+b=q+1$ and $q-b\ge b-a+1$};
\end{tikzpicture}
\mbox{}\\
\vskip 0.5cm
\begin{tikzpicture}[domain=0:11,x=1cm,y=1cm]
\draw[very thick,->] (-0.5,0) -- ++(6,0);
\draw[very thick,->] (0,-0.5) -- ++(0,3.5);
\node[right]at(5.5,0){$i$};
\node[above]at(0,3){$d^-_P(i)$};
\SOMMET{0,1}; \SOMMET{1.5,2.5}; \SOMMET{2.5,1.5}; \SOMMET{3.5,0.5}; \SOMMET{5,2}; 
\draw[very thick,dotted] (0,1)--(1.5,2.5);
\draw[very thick,dotted] (1.5,2.5)--(2.5,1.5);
\draw[very thick,dotted] (2.5,1.5)--(3.5,0.5);
\draw[very thick,dotted] (3.5,0.5)--(5,2);
\draw[thick,dashed] (1.5,0)--(1.5,2.5);
\draw[thick,dashed] (2.5,0)--(2.5,1.5);
\draw[thick,dashed] (3.5,0)--(3.5,0.5);
\draw[thick,dashed] (0,2.5)--(1.5,2.5);
\draw[thick,dashed] (0,0.5)--(3.5,0.5);
\node[below] at(1.5,0){$a-1$};
\node[below] at(2.5,0){$\p2$};
\node[below] at(3.5,0){$b$};
\node[left]at(0,0.5){$q-b$};
\node[left]at(0,1){$b-a+1$};
\node[left]at(0,2.5){$b$};
\node[below]at(2,-1){(b) $a+b=q+1$ and $q-b < b-a+1$};
\end{tikzpicture}
\caption{Indegree paths for the proof of Theorem~\ref{th:interval}, case $b-a+1<{\p2}$}
\label{fig:Nab-3}
\end{center}
\end{figure}


\item $b-a+1<{\p2}$.\\
The indegree path of~$P$
contains an ascent,
possibly a plateau (if $a\neq q+1-b$),
a descent, possibly a plateau and an ascent.
If $a\le q+1-b$, the first vertex of the first plateau is $a-1$
and its last vertex is $q-b$ (so that there is no plateau if $a=q+1-b$).
If $a>q+1-b$, the first vertex of the first plateau is $q-b$
and its last vertex is $a-1$.

We will prove that every vertex of~$P$ is fixed by every automorphism of~$P$.
Note that $d^-_P(0)=b-a+1$, so that $d^-_P(q)=q-b+a-1>d^-_P(0)$
since $b-a+1<{\p2}$.

We consider two subcases, depending on whether $\IS(P)$ contains a plateau or not.

\begin{enumerate}
\item $a=q+1-b$.\\
In that case, $\IS(P)$ does not contain any plateau,
the first vertex of the descent is $a-1$ and its last vertex is $b$.
Moreover, $d^-_P(a-1)=b$ and $d^-_P(b)=q-b$.
Let us consider the vertex $b$. 
If $d^-_P(b)=q-b\ge b-a+1=d^-_P(0)$  (see Figure~\ref{fig:Nab-3}(a)), then $|V_{q-b}(P)|=2$ 
and thus $b$ is fixed by every automorphism of~$P$.
If $q-b>b-a+1$ (see Figure~\ref{fig:Nab-3}(b)), then $|V_{q-b}(P)|=1$ and thus, again,
$b$ is fixed by every automorphism of~$P$.

Now, we claim that every set $V_d(P)$ with $|V_d(P)|=3$ contains
three fixed vertices. Each such set is of the form $\{i,j,k\}$, $i<j<k$,
with $j+k=2b$.
Then, $j$ and $k$ disagree on $b$ (since $b-j=k-b<q-b$, so that
$jb$ and $bk$ are arcs in $P$) and, as $b$
is fixed by every automorphism of~$P$, by Observation~\ref{obs:agree-fixpoint},
$i$, $j$ and $k$ are also fixed by every automorphism of~$P$.
The result then follows by Corollary~\ref{cor:half-is-fixed}.

\item $a\neq q+1-b$.\\
In that case, the vertex $a-1$ is either the first or the last vertex of
the first plateau and $d^-_P(a-1)=b$ or $d^-_P(a-1)=q-a+1$, respectively.
Similarly, the vertex $b$ is either the first or the last vertex of
the last plateau and $d^-_P(b)=q-b$ or $d^-_P(b)=a-1$, respectively.

We first claim that whenever 0 is fixed by every automorphism of~$P$,
all vertices belonging to a set $V_d(P)$ with $|V_d(P)|=3$ are also
fixed by every automorphism of~$P$.
Indeed, for any such set $\{i,j,k\}$
with $i<j<k$,
we have $i<a-1$ and $a-1<j<b$, so that $i$ and $j$ disagree
on 0. Therefore, $i$, $j$ and $k$ are all 
fixed by every automorphism of~$P$.

\begin{figure}
\begin{center}
\begin{tikzpicture}[domain=0:11,x=1cm,y=1cm]
\draw[very thick,->] (-0.5,0) -- ++(6,0);
\draw[very thick,->] (0,-0.5) -- ++(0,3.5);
\node[right]at(5.5,0){$i$};
\node[above]at(0,3){$d^-_P(i)$};
\SOMMET{0,0.5}; \SOMMET{1.5,2}; \SOMMET{2,2}; \SOMMET{2.5,1.5}; \SOMMET{3,1}; \SOMMET{3.5,1}; \SOMMET{4.5,2}; \SOMMET{5,2.5}; 
\draw[very thick,dotted] (0,0.5)--(1.5,2);
\draw[very thick,dotted] (1.5,2)--(2,2);
\draw[very thick,dotted] (2,2)--(3,1);
\draw[very thick,dotted] (3,1)--(3.5,1);
\draw[very thick,dotted] (3.5,1)--(5,2.5);
\draw[thick,dashed] (4.5,0)--(4.5,2);
\draw[thick,dashed] (1.5,0)--(1.5,2);
\draw[thick,dashed] (2.5,0)--(2.5,1.5);
\draw[thick,dashed] (3,0)--(3,1);
\draw[thick,dashed] (0,2)--(1.5,2);
\draw[thick,dashed] (2,2)--(4.5,2);
\draw[thick,dashed] (0,1)--(3,1);
\node[below] at(4.5,0){$x$};
\node[below] at(3,0){$b$};
\node[below] at(1.5,0){$a-1$};
\node[below] at(2.5,0){$\p2$};
\node[left]at(0,1){$q-b$};
\node[left]at(0,0.5){$b-a+1$};
\node[left]at(0,2){$b$};
\node[below]at(2,-1){(a) $a+b\neq q+1$ and $q-b > b-a+1$};
\end{tikzpicture}
\mbox{}\\
\vskip 0.5cm
\begin{tikzpicture}[domain=0:11,x=1cm,y=1cm]
\draw[very thick,->] (-0.5,0) -- ++(6,0);
\draw[very thick,->] (0,-0.5) -- ++(0,2.5);
\node[right]at(5.5,0){$i$};
\node[above]at(0,2){$d^-_P(i)$};
\SOMMET{0,0.5}; \SOMMET{1,1.5}; \SOMMET{2,1.5}; \SOMMET{2.5,1}; \SOMMET{3,0.5}; \SOMMET{3.5,0.5}; \SOMMET{4,0.5}; \SOMMET{5,1.5}; 
\draw[very thick,dotted] (0,0.5)--(1,1.5);
\draw[very thick,dotted] (1,1.5)--(2,1.5);
\draw[very thick,dotted] (2,1.5)--(3,0.5);
\draw[very thick,dotted] (3,0.5)--(4,0.5);
\draw[very thick,dotted] (4,0.5)--(5,1.5);
\draw[thick,dashed] (1,0)--(1,1.5);
\draw[thick,dashed] (2.5,0)--(2.5,1);
\draw[thick,dashed] (3,0)--(3,0.5);
\draw[thick,dashed] (3.5,0)--(3.5,0.5);
\draw[thick,dashed] (0,0.5)--(3,0.5);
\draw[thick,dashed] (0,1.5)--(1,1.5);
\node[below] at(2.5,0){$\p2$};
\node[below] at(1,0){$a-1$};
\node[below] at(3,0){$b$};
\node[below] at(3.5,0){$x$};
\node[left]at(0,0.5){$b-a+1$};
\node[left]at(0,1.5){$b$};
\node[below]at(2,-1){(b) $a+b\neq q+1$, $q-b = b-a+1$ and $a<q+1-b$};
\end{tikzpicture}
\mbox{}\\
\vskip 0.5cm
%
\begin{tikzpicture}[domain=0:11,x=1cm,y=1cm]
\draw[very thick,->] (-0.5,0) -- ++(7,0);
\draw[very thick,->] (0,-0.5) -- ++(0,3.5);
\node[right]at(6.5,0){$i$};
\node[above]at(0,3){$d^-_P(i)$};
\SOMMET{0,1}; \SOMMET{1.5,2.5}; \SOMMET{2,2.5}; \SOMMET{3,1.5}; \SOMMET{4,0.5}; \SOMMET{4.5,0.5}; \SOMMET{6,2}; 
\SOMMET{3.5,1};\SOMMET{5,1};
\draw[very thick,dotted] (0,1)--(1.5,2.5);
\draw[very thick,dotted] (1.5,2.5)--(2,2.5);
\draw[very thick,dotted] (2,2.5)--(4,0.5);
\draw[very thick,dotted] (4,0.5)--(4.5,0.5);
\draw[very thick,dotted] (4.5,0.5)--(6,2);
\draw[thick,dashed] (2,0)--(2,2.5);
\draw[thick,dashed] (3,0)--(3,1.5);
\draw[thick,dashed] (4.5,0)--(4.5,0.5);
\draw[thick,dashed] (3.5,0)--(3.5,1);
\draw[thick,dashed] (5,0)--(5,1);
\draw[thick,dashed] (0,2.5)--(1.5,2.5);
\draw[thick,dashed] (0,0.5)--(4,0.5);
\draw[thick,dashed] (0,1)--(5,1);
\node[below] at(3,0){$\p2$};
\node[below] at(2,0){$a-1$};
\node[below] at(4.5,0){$b$};
\node[below] at(3.5,0){$x$};
\node[below] at(5,0){$y$};
\node[left]at(0,0.5){$a-1$};
\node[left]at(0,1){$b-a+1$};
\node[left]at(0,2.5){$q-a+1$};
\node[below]at(2.5,-1){(c) $a+b\neq q+1$, $a-1 < b-a+1$ and $a>q+1-b$};
\end{tikzpicture}
\caption{Indegree paths for the proof of Theorem~\ref{th:interval}, case $b-a+1<{\p2}$ (Cont.)}
\label{fig:Nab-3bis}
\end{center}
\end{figure}


We now consider three subcases, depending on the values of $d^-_P(b)$ and $d^-_P(0)$.

\begin{enumerate}
\item $d^-_P(b)>d^-_P(0)$.\\
In that case (see Figure~\ref{fig:Nab-3bis}(a)),  the vertex 0
is fixed by every automorphism of~$P$, so that
all vertices belonging to a set $V_d(P)$ with $|V_d(P)|=3$ are also
fixed by every automorphism of~$P$.

If $a-1$ is the first vertex of the first plateau, let $x>\p2$ be the unique vertex with $d^-_P(x)=d^-_P(a-1)$.
We then have $V_b(P)=\{a-1,a,\dots,q-b,x\}$.
Since $a-1$ and $x$ are the only in-neighbours of 0 in $V_b(P)$,
both of them are fixed by every automorphism of~$P$.
Therefore, by Lemma~\ref{lem:rigid-plateau},
every vertex belonging to the first plateau is
fixed by every automorphism of~$P$.

Suppose now that $a-1$ is the last vertex of the first plateau.
Similarly, let $x<\p2$ be the unique vertex with $d^-_P(x)=d^-_P(b)$,
so that $V_{a-1}(P)=\{x,q-a+1,\dots,b\}$.
Since $x$ is the unique out-neighbour of~0 in $V_{a-1}(P)$
(as $x<a$ and $q-a+1>a$, which gives $V_{a-1}(P)\ \cap\ [a,b]=\{q-a+1,\dots,b\}$), 
$x$ is fixed by every automorphism of~$P$ and,
by Proposition~\ref{prop:fixed-by-symmetry}, $q-x$ is also fixed by every automorphism of~$P$.
Therefore, by Lemma~\ref{lem:rigid-plateau},
every vertex belonging to the first plateau is
fixed by every automorphism of~$P$.

The result then follows by Corollary~\ref{cor:half-is-fixed}.

\item $d^-_P(b)= d^-_P(0)$.\\
Suppose first that $a-1$ is the last vertex of the first plateau (that is, $a>q+1-b$), so that $b$ is
the last vertex of the second plateau.
In that case, the vertices of the second plateau induce a transitive
tournament and they all agree on 0. Therefore, the sub-tournament induced
by $V_{b-a+1}(P)$ is rigid, so that all its vertices, and in particular 0,
are fixed by every automorphism of~$P$.
Hence,
all vertices belonging to a set $V_d(P)$ with $|V_d(P)|=3$ are also
fixed by every automorphism of~$P$.
By Proposition~\ref{prop:fixed-by-symmetry}, all vertices of the first plateau
are also fixed by every automorphism of~$P$ and, again,
the result follows by Corollary~\ref{cor:half-is-fixed}.

\medskip

Suppose now that $a-1$ is the first vertex of the first plateau
(that is, $a<q+1-b$), so that $b$ is
the first vertex of the second plateau (see Figure~\ref{fig:Nab-3bis}(b)).

We first prove that the vertex 0 is fixed by every automorphism of~$P$.
Suppose to the contrary that 0 belongs
to an orbit of size at least~3 of some automorphism $\phi$ of~$P$. Among the vertices of the second
plateau, $b$ is the only in-neighbour of 0. Therefore, this orbit must be of 
size~3 (recall that, by Observation~\ref{obs:orbit-regular}, the sub-tournament induced by an orbit must be regular)
and thus a 3-cycle $b0x$. Note that since $xb$ is an arc of~$P$,
we then have $x-b\ge a$.
Consider now the set $V_{b-a+2}=\{1,b-1,q-a+2\}$.
The vertex 0 has only one in-neighbour in $V_{b-a+2}$, namely $b-1$,
since $b-1\in [a,b]$, $1\notin [a,b]$
and $q-a+2\notin [a,b]$ (since $q-a+2>b$).
On the other hand, $b$ has at least two in-neighbours in $V_{b-a+2}$, namely
$1$ (since $b-1\in [a,b]$) and $q-a+2$ (since $q-a+2-b>x-b\ge a$
and $q-a+2-b<\p2<b$).
By Observation~\ref{obs:agree-odd-orbits}, this implies that $V_{b-a+2}$ cannot be an orbit of $\phi$,
so that, in particular, $1$ is fixed by $\phi$.
Now, note that $1$ is an in-neighbour of $x$ (since
$x-1\ge b+a-1>b$) and an out-neighbour of 0,
a contradiction since, by Observation~\ref{obs:agree-fixpoint}, all vertices in $\{0,b,x\}$ should agree
on every vertex fixed by~$\phi$.

Therefore, 0 is fixed by every automorphism of~$P$, so that
all vertices belonging to a set $V_d(P)$ with $|V_d(P)|=3$ are also
fixed by every automorphism of~$P$.
By Proposition~\ref{prop:fixed-by-symmetry}, $q$ is also
fixed  by every automorphism of~$P$ and thus, by
 Lemma~\ref{lem:rigid-plateau},
 every vertex of the first plateau is fixed  by every automorphism of~$P$.
 The result then follows by Corollary~\ref{cor:half-is-fixed}.

\item $d^-_P(b)< d^-_P(0)$.\\
Suppose first that $a-1$ is the last vertex of the first plateau (that is, $a>q+1-b$), so that $b$ is
the last vertex of the second plateau (see Figure~\ref{fig:Nab-3bis}(c)).
In that case, the vertices of each plateau induce a transitive
tournament and thus all vertices of these plateaus are fixed  by every automorphism of~$P$.
Consider the set $V_{b-a+1}=\{0,x,y\}=\{0,q-b+a-1,2b-2a+2\}$.
Vertices 0 and $q-b+a-1$ disagree on $q-a+1$
(since $q-a+1\in [a,b]$
and $r=q-a+1-(q-b+a-1)<q-b<a-1$ so that $r\notin [a,b]$), 
which is fixed  by every automorphism of~$P$,
and thus $0$, $q-b+a-1$ and $2b-2a+2$ are all fixed by every automorphism of~$P$.

Therefore,
all vertices belonging to a set $V_d(P)$ with $|V_d(P)|=3$ are also
fixed by every automorphism of~$P$, and the 
result follows by Corollary~\ref{cor:half-is-fixed}.

\medskip

Suppose now that $a-1$ is the first vertex of the first plateau (that is, $a<q+1-b$), so that $b$ is
the first vertex of the second plateau.

We first prove that the vertex 0 is fixed by every automorphism of~$P$.
Suppose to the contrary that 0 belongs
to an orbit of size~3 of some automorphism $\phi$ of~$P$, say a 3-cycle 
$0xy$, with $a<x<b$ and $y>b+1$,
so that $x0$, $0y$ and $yx$ are arcs in $P$, implying in particular $a\le y-i\le b$.
Consider now the set $V_{b-a+2}=\{1,x-1,y+1\}$.
Since $1<a$, $a\le x-2$ (as $a-1<\p2<x$) and $x<b$, we get that $01$, $x1$, $(x-1)0$ and $(x-1)x$
are all arcs in $P$, so that $1$ and $x-1$ disagree on the orbit of 0 with respect to $\phi$.
This implies that $V_{b-a+2}$ cannot be an orbit of $\phi$,
so that, in particular, $1$ is fixed by $\phi$.
Now, note that $1$ is an in-neighbour of $y$ (since $y-1>b$)
and an out-neighbour of 0,
a contradiction since all vertices in $\{0,x,y\}$ should agree
on every vertex fixed by~$\phi$.

Therefore, 0 is fixed by every automorphism of~$P$, so that
all vertices belonging to a set $V_d(P)$ with $|V_d(P)|=3$ are also
fixed by every automorphism of~$P$.
By Lemma~\ref{lem:rigid-plateau},
every vertex of the first plateau is fixed  by every automorphism of~$P$
and the result follows by Corollary~\ref{cor:half-is-fixed}.

\end{enumerate}

\end{enumerate}

\end{enumerate}

This concludes the proof.
\end{proof}

Again, since any tournament $T$ is isomorphic to its converse, Theorem~\ref{th:interval}
gives the following.

\begin{corollary}
If $T=T(2q+1;S)$ is a cyclic tournament, such that $S^-=[1,a]\cup[b,q]$, $1\le a\le b\le q$,
then $T$ satisfies the Albertson-Collins Conjecture.
\label{cor:interval}
\end{corollary}

\subsection{Paley tournaments}
\label{subsec:paley}

Let $p$ be a prime, $p\equiv 3\pmod 4$, so that $p=2q+1$ with $q$ odd.
The \emph{Paley tournament} is the cyclic tournament 
$QR_p=T(p;S_p)$ where $S_p$ is the set of non-zero quadratic residues of $p$.
That is, for every $i,j$, $i<j$, $ij$ is an arc in $QR_p$  if and only if $j-i$ is a non-zero square in $\ZZZ_p$.
Paley tournaments are arc-transitive and the set $\Aut(QR_p)$
is the set of mappings $i\mapsto ai+b$, where $a,b\in\ZZZ_p$ and $a$ is a non-zero square.
Since the number of non-zero squares in $\ZZZ_p$ is $q=\left\lfloor\frac{p}{2}\right\rfloor$,
we thus have $|\Aut(QR_p)|=pq$.

In~\cite{AC99}, Albertson and Collins show that the Paley tournament $QR_7$ satisfies
their conjecture. We prove here that all Paley
tournaments satisfy the Albertson-Collins Conjecture.

\begin{theorem}
For every prime $p$, $p\equiv 3\pmod 4$, $QR_p$ satisfies the Albertson-Collins Conjecture.
\label{th:paley}
\end{theorem} 

\begin{proof}
Suppose to the contrary that there exists a nontrivial $\lambda^*$-preserving
automorphism $\phi:i\mapsto ai+b$, 
where $a,b\in\ZZZ_p$ and $a$ is a non-zero quadratic residue, and let $\ell$
denote the order of $\phi$.

Consider the action on $V(QR_p)$ of the group
$H=\,\,<\phi>\,\,=\{\phi^k:\ 1\le k\le\ell\}$.
Note that $|H|$ divides $|\Aut(QR_p)|=pq$
and, since $\phi$ is $\lambda^*$-preserving, $|H|$ divides $q$.
Moreover, $q$ cannot be prime, since otherwise the sub-tournament $T_2$ of $QR_p$ induced by the set of vertices
$\{q+1,\dots,2q\}$ would have a single orbit with respect to $\phi$, in contradiction
with Proposition~\ref{prop:no-single-orbit}.

We claim that we necessarily have $b=0$. Assume to the contrary that $b\neq 0$ and
let $r$ be the size of the orbit of 0 with respect to $\phi$.
We have $\phi(0)=b$, $\phi^2(0)=ab+b$ and so on, so that
$$\phi^r(0)=a^{r-1}b+a^{r-2}b+\dots+ab+b=0$$
and $\phi^{r+1}(0)=\phi(0)=b$.
On one other hand, we have
$$\phi^{r+1}(0)=a^{r}b+a^{r-1}b+\dots+ab+b=a^rb.$$
Therefore, $a^rb=b$ and thus $a^r=1$.
Since $a$ is a non-zero quadratic residue in $\ZZZ_p$, its order 
in the multiplicative group $\ZZZ_p^*$ is $q$ and thus $q\,|\,r$. 
Moreover, since $\phi$ is $\lambda^*$-preserving, we have $r=q$.
Hence, the sub-tournament $T_1$ of $QR_p$ induced by the set of vertices
$\{0,\dots,q\}$ contains a fixed vertex $x$ such that all vertices of the orbit of~0
with respect to $\phi$
agree on $x$, which implies either $d^+_{T_1}(x)=0$ or $d^+_{T_1}(x)=q$
in contradiction with the definition of $QR_p$ since both sets
$\{1,\dots,q\}$ and $\{q+1,\dots,2q\}$ contain non-zero quadratic residues.

We thus have $b=0$, so that 0 is fixed by $\phi$.
Moreover, the orbit of 1 with respect to $\phi$ is
$\{1,a,a^2,\dots,a^{q-1}\}$ (recall that the order of $a$
in the multiplicative group $\ZZZ_p^*$ is $q$) and its size is~$q$.
Since $\phi$ is $\lambda^*$-preserving, the orbit of 1 with respect to $\phi$ is thus
$\{1,\dots,q\}$. Therefore, all vertices in $\{1,\dots,q\}$
agree on 0 and we obtain a contradiction as above.
\end{proof}

\section{Discussion}
\label{sec:discussion}

In this paper, we proposed several sufficient conditions for a cyclic tournament
to satisfy the Albertson-Collins Conjecture.
Proposition~\ref{prop:T1orT2isrigid} says that the cyclic tournament $T=T(2q+1;S)$
satisfies the Albertson-Collins Conjecture whenever any of the sub-tournaments $T_{0,q}$
or $T_{q+1,2q}$ is rigid (however, this condition is not necessary, as shown by Example~\ref{ex:T13}).
Using this property, we then got several sufficient conditions on the structure of the
set $S^-$ of negative connectors for at least one of these two sub-tournaments to be rigid.

We finally propose some directions for future research.
\begin{enumerate}
\item The main question following our work is obviously to prove, or disprove,
the Albertson-Collins Conjecture.

\item In~\cite{B06}, Boutin introduces the notion of a \emph{determining set}
of a graph $G$, that is,
a subset $X$ of $V(G)$ such that, for every two automorphisms $\phi_1,\phi_2\in \Aut(G)$,
$\phi_1=\phi_2$ whenever $\phi_1(x)=\phi_2(x)$ for every vertex $x\in X$
(this notion was independently introduced by Erwin and Harary in~\cite{EH06}, 
under the name of \emph{fixing set}).
Albertson and Boutin prove in~\cite{AB07} that $D(G)\le d$ if and only if $G$ has
a determining set $X$ with $D(X)\le d-1$.
Therefore $D(G)\le 2$ if and only if the subgraph of $G$ induced by $X$ is rigid.

For each cyclic tournament $T$ of order $2q+1$ for which we proved the Albertson-Collins
Conjecture, we exhibited a rigid determining set of size at most 
$\left\lceil\p2\right\rceil$.
It is not difficult to see that every Paley tournament has a rigid determining set of size 2
(indeed, every pair of vertices is a rigid determining set).
In~\cite{L13}, Lozano proves that every tournament contains a (not necessarily rigid)
determining set of size $\left\lfloor\frac{p}{3}\right\rfloor$.

It would thus be interesting to determine the minimal size of a rigid determining
set in a cyclic tournament. 

\item In~\cite{B13a,B13b}, Boutin also defines the \emph{cost of distinguishing
a graph $G$} such that $D(G)=2$, denoted $\rho(G)$,
 as the minimum size of a label class in a distinguishing 2-labeling of $G$.
 It is easy to prove that $\rho(QR_p)=2$ for every Paley tournament
 $QR_p$ (indeed, for every two vertices $i$ and $j$ in $QR_p$, the 2-labeling 
 $\lambda_{i,j}$ defined by $\lambda_{i,j}(v)=1$ if and only if $v=i$ or $v=j$,
 is distinguishing). 
 
 It would  be interesting to determine $\rho(T)$ for every
 cyclic tournament $T=T(2q+1;S)$,
 or to characterize cyclic tournaments $T$ with $\rho(T)=2$.

\end{enumerate}


%


\end{document}